\documentclass[twocolumn]{autart}
\usepackage[utf8]{inputenc}
\usepackage{mathtools, cuted}
\usepackage{amssymb}
\usepackage{balance}
\usepackage{mathrsfs}
\usepackage{comment}
\usepackage{enumitem}
\usepackage{mathdots}
\usepackage{xfrac}
\usepackage{bm}
\usepackage{array}
\usepackage{color}
\usepackage{algorithm}
\usepackage{algorithmicx}
\usepackage{algpseudocode}
\usepackage{physics}
\usepackage{lipsum}
\usepackage{float}

\newtheorem{theorem}{Theorem}
\newtheorem{lemma}{Lemma}
\newtheorem{proposition}{Proposition}
\newtheorem{corollary}{Corollary}

\usepackage{bbm}                         
\newcommand{\zon}[2]{\langle #1, #2\rangle}
\usepackage{tikz}

\usepackage{pgfplots}
\usepgflibrary{shapes}
\usepackage{verbatim}
\usepackage{caption}
\usepackage{subcaption}

\usetikzlibrary{arrows.meta,positioning,shapes.geometric,calc}
\usepackage{amsmath,amssymb,bm}

\tikzstyle{Rect}=[draw=blue,line width=0.001pt,preaction={clip, postaction={pattern=north east lines, pattern color=blue,line width=0.1pt}}]
\tikzset{
	>=Stealth,
	help lines/.style={dashed, thick},
	axis/.style={<->},
	important line/.style={thick},
	connection/.style={thick, dotted},
}

\newcommand{\redu}[2]{\downarrow_{#1}\hspace{-0.06cm}#2}

\DeclareMathAlphabet\mathbfcal{OMS}{cmsy}{b}{n}

\newtheorem{definition}{Definition}
\usepackage{etoolbox}

\makeatletter
\def\@xnamedef#1{\expandafter\protected@xdef\csname #1\endcsname}
\def\no@harm{} 
\def\ead@au#1{\protected@edef\@ead@au{#1}}
\patchcmd\runningauthor@fmt{\global\edef}{\protected@xdef}{}{}
\patchcmd\runningauthor@fmt{\global\edef}{\protected@xdef}{}{}
\patchcmd\author@fmt{\edef}{\protected@edef}{}{}
\patchcmd\add@xtok{\xdef}{\protected@xdef}{}{}
\makeatother

\usepackage{eso-pic}
\AddToShipoutPictureBG*{%
	\AtPageUpperLeft{%
		\setlength\unitlength{1in}%
		\hspace*{\dimexpr0.5\paperwidth\relax}
		\makebox(8,-1.3)[c]{
			\begin{tabular}{c c}
				Rodrigo A. Gonz\'alez \emph{et al.},
				The Zonotopic Mixture Filter, \\
				uploaded to arXiv on Aug 18th, 2026. \\
\end{tabular}}}}

\begin{document}
\begin{frontmatter}

\title{The Zonotopic Mixture Filter\thanksref{footnoteinfo}} 

\thanks[footnoteinfo]{This paper was not presented at any IFAC 
meeting. Corresponding author: Rodrigo A. Gonz\'alez.}

\author[First]{Rodrigo A. González},
\ead{r.a.gonzalez@tue.nl}
\author[Second]{Angel L. Cedeño},
\ead{angel.cedeno@usach.cl}
\author[Third]{Vicen\c{c} Puig}
\ead{vicenc.puig@upc.edu}

\address[First]{Department of Mechanical Engineering, Eindhoven University of Technology, Eindhoven, The Netherlands}
\address[Second]{Electrical Engineering Department, University of Santiago of Chile, Santiago, Chile}
\address[Third]{Institut de Rob\`{o}tica i Inform\`{a}tica Industrial (CSIC-UPC), Universitat Polit\`{e}cnica de Catalunya, Barcelona, Spain}

\begin{abstract}
State estimation is commonly posed in either a probabilistic or an unknown-but-bounded framework. The former requires a fully specified noise distribution, typically with unbounded support, while the latter yields guaranteed enclosures that carry no probabilistic weighting. Bridging these noise descriptions, this paper proposes a zonotopic mixture noise model, in which the noise is generated by drawing a zonotope from a finite collection according to fixed probabilities and then realizing an arbitrary element of it. For this noise model, we derive the zonotopic mixture filter, which propagates a bank of zonotopic Kalman filters over mode histories, discards the histories falsified by the data, and weights the surviving ones by their relative probability. The resulting state enclosures yield guaranteed coverage probabilities and remain valid for every noise realization compatible with the bounds, and a greedy mixture reduction scheme preserves these statistical guarantees while keeping the representation tractable. Numerical examples illustrate the proposed approach and its potential benefits over related state estimation methods.
\end{abstract}


\end{frontmatter}

\section{Introduction}
State estimation, the process of inferring the internal variables of a dynamical system from data, is a key tool in filtering, identification, and control \cite{anderson1979}. Ubiquitous in model predictive control \cite{rawlings2017model}, it also plays a central role in prognosis and fault detection \cite{almohamad2021prognosis}, localization and sensor fusion \cite{wang2018zonotopic,zhang2025distributed}, and system identification \cite{gibson2005,schon2011,Cedeno2024Id}.

Two frameworks for state estimation are prevalent in practice, distinguished by how the noise is modeled: the probabilistic setting and the unknown-but-bounded (UBB) setting. In the probabilistic setting, noise is a stochastic process with a known distribution, and state estimation follows from Bayesian filtering recursions \cite{sarkka2013bayesian}. For linear systems with additive Gaussian noise, these recursions admit the closed-form solution given by the Kalman filter \cite{kalman1960}, and generalize to the Gaussian sum filter for Gaussian mixture noise \cite{alspach1972nonlinear,kitagawa1987non} or for static nonlinearities \cite{arasaratnam2007discrete,Cedeno2024quadrature}. For arbitrary densities and nonlinear dynamics, the particle filter approximates the recursions through a weighted sample representation \cite{gordon1993novel}. Because the noise distribution is fully specified, the recursions return a point estimate together with its covariance, which is useful for uncertainty quantification. However, such a specification is often inaccurate or an oversimplification of the true noise characteristics, and the densities in use typically have unbounded support, even though disturbances are physically bounded in practice.

The UBB setting instead assumes only that the noise belongs to a bounded set, with no probabilistic structure on its realization. The set-membership estimators provide a set of state estimates, all of which are consistent with the collected data. Algorithms in this domain differ in the geometry adopted for noise bounds, such as polytopes, ellipsoids, or zonotopes \cite{depaula2022zonotopic}. Zonotopes, which are affine images of unit hypercubes, are closed under the linear maps and Minkowski sums appearing in the state recursion, which makes them a convenient modeling choice. They have been considered in \cite{alamo2005guaranteed}, and in the zonotopic Kalman filter (ZKF) of \cite{combastel2015zonotopes}, which later extended to nonlinear \cite{rego2020guaranteed} and switched \cite{ifqir2022zonotopic} systems. However, set-membership estimates are conservative by construction: every realization within the bounds is treated as admissible, and no mechanism distinguishes likely from unlikely trajectories. This is limiting in safety-critical applications, where probabilistic information is essential to assess reliability and performance~\cite{abate2008probabilistic,mesbah2016stochastic}.

The shortcomings of both approaches have motivated combinations of the two settings. In \cite{combastel2016extended}, the noise is decomposed into the sum of a zonotopically bounded term and a Gaussian term. The same decomposition is used in \cite{zhan2023zonotope}, extended to nonlinear systems in \cite{de2024extended}, and applied to data fusion in \cite{salhi2024zonotopic}.  However, the unbounded Gaussian component means that the state is enclosed only up to a prescribed confidence level rather than with certainty, and the split of the noise into UBB and stochastic terms can be difficult to motivate from a physical standpoint. A different route maintains a purely stochastic description with bounded support, e.g., uniform or truncated-Gaussian noise \cite{pavelkova2018approximate,gonzalez2025truncated}; however, as in any probabilistic approach, a specific distribution must be assumed. Alternatively, \cite{benavoli2016probabilistic} shows that pure set-membership filtering admits an equivalent probabilistic reformulation through sets of probability measures, although the resulting description does not carry probability weighting between different noise regions.

This paper proposes an alternative merge of the probabilistic and UBB paradigms in state estimation by introducing a the zonotopic mixture model for the measurement and output noise. At each time instant, this model generates a noise realization by first selecting one zonotope from a finite collection according to fixed mode probabilities, and then selecting an arbitrary element of such zonotope. This is the set-membership counterpart of a finite mixture model and recovers the UBB zonotopic description and point-mass models \cite{bergman1999recursive,giraldo2024gaussian,matousek2026survey} as special cases. Since no in-mode distribution is assumed, the model is consistent with every noise law supported on the given zonotopes. Unlike box particle filtering \cite{gning2013introduction}, where weighted boxes discretize the posterior of a fully stochastic model, the zonotopes here define the noise model itself. Because they may overlap, the mode probabilities refine a bounded region into subregions of different likelihood, and retaining only the highest-weight components of the mixture yields enclosures that are tighter than the worst-case set and valid with a guaranteed probability.

The main contributions of this paper are as follows:
\begin{enumerate}[label=C\arabic*,leftmargin=*]
    \item \label{contribution2} We derive the zonotopic mixture filter (ZMF, Theorem~\ref{thm:zmf}) for a dynamical system subject to zonotopic mixture noise, and show that it delivers state enclosures with guaranteed coverage probabilities (Corollary~\ref{cor:prob}), valid for every noise realization compatible with the bounds.
    \item \label{contribution3} We analyze the connection between the ZMF, the Gaussian sum filter \cite{alspach1972nonlinear} and an exact discrete Bayes filter. In the small-dispersion regime, the ZMF coincides with the latter below a positive threshold, whereas the Gaussian sum filter is equivalent to the ZMF only in the limit (Theorem~\ref{thm:filters-dbf}).
    \item \label{contribution4} We formalize admissible zonotopic mixture reductions (Definition~\ref{def:zm_reduction}), which preserve the enclosure guarantees, and propose a greedy algorithm based on a closed-form merge cost (Lemma~\ref{lem:cost_closed_form}, Algorithm~\ref{alg:zm_reduction}) that keeps the ZMF computationally tractable.
\end{enumerate}

The remainder of this paper is organized as follows. Section~\ref{sec:setup} introduces the notation, the system model, and the filtering problem. Section~\ref{sec:zkf} recaps the zonotopic Kalman filter, and Section~\ref{sec:zmf} presents the main contribution, the zonotopic mixture filter. The zonotopic mixture reduction is developed in Section~\ref{sec:implementation}, simulation studies are reported in Section~\ref{sec:simulation}, and Section~\ref{sec:conclusions} concludes the paper.

\section{System description and problem formulation}
\label{sec:setup}

\subsection{Notation and preliminaries}
\label{subsec:notation}

All vectors and matrices are written in bold and vectors are column vectors, unless transposed. The quantity $\mathbf{x}_{1:t}$ refers to the sequence of vectors $\{\mathbf{x}_1,\mathbf{x}_2,\dots,\mathbf{x}_t\}$. The expression $\delta(\cdot,\cdot)$ denotes the Kronecker delta function, where $\delta(x,y)=1$ if $x=y$, and $\delta(x,y)=0$ otherwise. Similarly, the indicator function $\mathbbm{1}\{\mathbf{x}\in \mathcal{A}\}=1$ if $\mathbf{x}\in \mathcal{A}$, and is equal to zero otherwise. If $\mathbf{A}$ is a matrix, $\|\mathbf{A}\|_2$ denotes the $L_2$ norm, and $\|\mathbf{A}\|_{2,1}$ denotes the $L_{2,1}$ norm, i.e., the sum of the Euclidean norm of the columns of $\mathbf{A}$. Given a symmetric positive definite matrix $\mathbf{W}\in\mathbb{R}^{n\times n}$, the weighted Euclidean norm of a vector $\mathbf{x}\in\mathbb{R}^{n}$ is defined as $\|\mathbf{x}\|_{\mathbf{W}}:=\sqrt{\mathbf{x}^\top \mathbf{W}\mathbf{x}}$. The probability measure is denoted by $\mathbb{P}\{\cdot\}$, $P(\cdot)$ is used for probability mass functions, and $p(\cdot)$ denotes a probability density function.

A zonotope $\langle \mathbf{c}, \mathbf{E} \rangle$ with center $\mathbf{c}\in \mathbb{R}^n$ and generator matrix $\mathbf{E}\in\mathbb{R}^{n\times p}$ is a convex polytope defined as the linear image of the unit hypercube $[-1,1]^p$ by $\mathbf{E}$, translated by $\mathbf{c}$: $\langle \mathbf{c}, \mathbf{E} \rangle=\{\mathbf{c} +\mathbf{Es},\|\mathbf{s}\|_{\infty}\leq 1\}$. Its weighted Frobenius radius, also denoted $F_\mathbf{W}$-radius, is $\sqrt{\textnormal{tr}\{\mathbf{W}\mathbf{E}\mathbf{E}^\top \}}$. The following linear map and Minkowski sum operations hold true on zonotopes:
\begin{align}
    \mathbf{A}\otimes \langle \mathbf{c}, \mathbf{E} \rangle &= \langle \mathbf{A}\mathbf{c}, \mathbf{A}\mathbf{E} \rangle, \notag \\
    \langle \mathbf{c}_1, \mathbf{E}_1 \rangle \oplus \langle \mathbf{c}_2, \mathbf{E}_2 \rangle &= \langle \mathbf{c}_1+\mathbf{c}_2, [\mathbf{E}_1,\mathbf{E}_2]\rangle. \notag
\end{align}
Given a matrix $\mathbf{X}\in\mathbb{R}^{n\times p}$, the row-sum matrix $\textnormal{rs}(\mathbf{X})$ is a diagonal matrix with diagonal entries given by $\sum_{j=1}^p |X_{ij}|$. Since the Minkowski sum increases the generator count, a reduction operator $\redu{q,\mathbf{W}}$ is needed to keep the representation tractable. Given a target $q\geq n$, sort the columns of $\mathbf{E}$ by decreasing weighted Euclidean norm; let $\mathbf{E}_{>}$ collect the $q-n$ largest and $\mathbf{E}_{<}$ collect the remainder. The reduced generator matrix
\begin{equation}
\label{eq:reducedgenerator}
    \redu{q,\mathbf{W}}{\mathbf{E}} := [\mathbf{E}_{>},\textnormal{rs}(\mathbf{E}_{<})] 
\end{equation}
satisfies the inclusion $\langle\mathbf{c},\mathbf{E}\rangle\subseteq \langle\mathbf{c},\redu{q,\mathbf{W}}{\mathbf{E}}\rangle$ \cite{Combastel2005}.

\subsection{System model}
\label{subsec:system}
Consider the discrete-time linear system in state-space form
\begin{subequations}
\label{eqn:system}
\begin{align}
    \label{eqn:state}
    \mathbf{x}_{t+1} &= \mathbf{A}\mathbf{x}_t + \mathbf{B}\mathbf{u}_t + \mathbf{w}_t, \\
    \label{eqn:output}
    \mathbf{y}_t &= \mathbf{C}\mathbf{x}_t + \mathbf{D}\mathbf{u}_t + \mathbf{v}_t,
\end{align}
\end{subequations}
where $\mathbf{x}_t\in\mathbb{R}^{n_x}$, $\mathbf{u}_t\in\mathbb{R}^{n_u}$, and $\mathbf{y}_t\in\mathbb{R}^{n_y}$ are the state, input, and output vectors, respectively. The system matrices $\mathbf{A}\in\mathbb{R}^{n_x\times n_x}$, $\mathbf{B}\in\mathbb{R}^{n_x\times n_u}$, $\mathbf{C}\in\mathbb{R}^{n_y\times n_x}$, and $\mathbf{D}\in\mathbb{R}^{n_y\times n_u}$ are taken constant in time to avoid cluttering the notation; the extension of our algorithms to time-varying $\mathbf{A}_t,\mathbf{B}_t,\mathbf{C}_t,\mathbf{D}_t$ is direct. The initial condition $\mathbf{x}_1$ is contained in the zonotope $\langle \bar{\mathbf{x}}_1, \mathbf{E}_1 \rangle$ for known $\bar{\mathbf{x}}_1$ and $\mathbf{E}_1$, and the input signal $\{\mathbf{u}_t\}$ is a fixed and deterministic sequence.

The process and measurement noises are described by the following noise model, which, to the best of the authors' knowledge, has not been considered previously in the state estimation literature.
 
\begin{definition}[Zonotopic mixture model]
\label{def:zonotopic_mixture}
Consider the collection of zonotopes $\{\langle \mathbf{c}_i, \mathbf{E}_i \rangle\}_{i=1}^{M}$ in $\mathbb{R}^{n}$ with $\mathbf{E}_i\in\mathbb{R}^{n\times p_i}$, and let $\{\alpha_i\}_{i=1}^{M}$ be positive weights satisfying $\sum_{i=1}^{M}\alpha_i=1$. A random vector $\mathbf{w}\in\mathbb{R}^n$ is said to follow the \emph{zonotopic mixture} $\{(\alpha_i,\langle \mathbf{c}_i, \mathbf{E}_i \rangle)\}_{i=1}^{M}$ if it admits the representation
\begin{equation}
    \mathbf{w} = \mathbf{c}_{m} + \mathbf{E}_{m}\mathbf{s}_m,
    \label{eqn:zm_realization}
\end{equation}
where the mode index $m$ is a random variable that takes values in $\{1,\dots,M\}$ with $\mathbb{P}\{m=i\}=\alpha_i$, and $\mathbf{s}_m$ is a vector that satisfies $\|\mathbf{s}_m\|_\infty\leq 1$, to which no probabilistic description is assigned a priori.
\end{definition}

The mode index is the only source of randomness in the description of the model, since conditioned on $m=i$, the noise is an UBB element $\mathbf{c}_i+\mathbf{E}_i\mathbf{s}_i$ of $\langle \mathbf{c}_i, \mathbf{E}_i \rangle$. The zonotopes need not cover a connected region, and they may have nonempty pairwise intersections. In agreement with the set-membership literature \cite{le2013zonotopes}, the noise may be generated stochastically within each zonotope, but  such distributional information is not assumed or exploited. Consequently, all results in this paper hold for every noise realization compatible with the bounds and, in particular, under any in-mode noise distribution.

\begin{figure}[t]
  \centering
  \includegraphics[width=0.98\columnwidth]{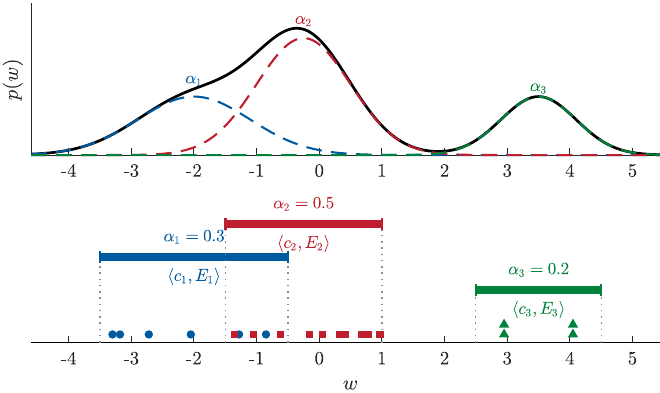}
  \vspace{-0.2cm}
    \caption{Two-stage generation mechanism shared by the Gaussian mixture and the zonotopic mixture. Top: Gaussian mixture with $M=3$ Gaussian components. Bottom: zonotopic mixture with $M=3$ components, where mode $i$ is selected with probability $\alpha_i$, after which $w$ is only known to lie in the zonotope $\langle \mathbf{c}_i, \mathbf{E}_i \rangle$, depicted as a horizontal bar (vertical placement is for visibility only). The blue dots, red squares, and green triangles show sets of noise realizations. Within each zonotope no in-mode distribution is assumed.}
  \label{fig:zmm_illustration}
\end{figure}

\begin{rem}
\label{rem:gmm_relation}
Definition~\ref{def:zonotopic_mixture} is a set-membership counterpart of a finite mixture model, of which the Gaussian mixture is the most widely used \cite{sorenson1971recursive}; Figure~\ref{fig:zmm_illustration} illustrates the analogy. Both descriptions share the same two-stage mechanism: a mode index is drawn with probability $\alpha_i$, after which the vector is generated by the mode-conditional component. In a Gaussian mixture, this component is fully specified by the density $\mathcal{N}(\mathbf{w};\bm{\mu}_i,\bm{\Sigma}_i)$; in the zonotopic mixture, the density is replaced by the set $\langle \mathbf{c}_i, \mathbf{E}_i \rangle$, and the only statement made about the in-mode behavior is membership in this set. A single zonotopic mixture is simultaneously consistent with every noise law whose modes are supported on the respective zonotopes. Two common noise models arise as degenerate cases. For $M=1$, the zonotopic mixture reduces to the UBB zonotopic noise description, while for $\mathbf{E}_i=\mathbf{0}$ for all $i$, the vector $\mathbf{w}$ follows a discrete distribution with atoms $\{\mathbf{c}_i\}_{i=1}^M$ and weights $\{\alpha_i\}_{i=1}^M$.
\end{rem}

For each $t\in\mathbb{N}$, the process noise $\mathbf{w}_t$ in \eqref{eqn:state} and the output noise $\mathbf{v}_t$ in \eqref{eqn:output} follow the zonotopic mixtures $\{(\alpha_i,\langle \bm{\mu}_i, \mathbf{Q}_i \rangle)\}_{i=1}^{W}$ and $\{(\eta_j,\langle \bm{\rho}_j, \mathbf{R}_j \rangle)\}_{j=1}^{V}$, respectively, with $\mathbf{Q}_i\in\mathbb{R}^{n_x\times n_w}$ and $\mathbf{R}_j\in\mathbb{R}^{n_y\times n_v}$ such that $\mathbf{R}_j\mathbf{R}_j^\top\succ \mathbf{0}$. The associated mode indices are denoted $m_t\in \mathcal{I}_W:=\{1,\dots,W\}$ and $n_t\in \mathcal{I}_V:=\{1,\dots,V\}$, so that $\mathbf{w}_t\in\langle \bm{\mu}_{m_t}, \mathbf{Q}_{m_t} \rangle$ and $\mathbf{v}_t\in\langle \bm{\rho}_{n_t}, \mathbf{R}_{n_t} \rangle$. The sequences $\{m_t\}$ and $\{n_t\}$ are i.i.d. in $t$, mutually independent, not known in advance, and do not depend on $\mathbf{x}_1$.
 
\subsection{Problem statement}
\label{subsec:problem}
Given input-output data $\{\mathbf{u}_{1:N},\mathbf{y}_{1:N}\}$ from the system in \eqref{eqn:system}, the known system parameters $(\mathbf{A},\mathbf{B},\mathbf{C},\mathbf{D})$, and the zonotopic mixture parameters $\{\alpha_i,\bm{\mu}_i,\mathbf{Q}_i\}_{i=1}^W$, $\{\eta_j,\bm{\rho}_j,\mathbf{R}_j\}_{j=1}^V$, the following problems are addressed:
 
\begin{enumerate}[label=P\arabic*,leftmargin=*]
\item \label{prob:enclosure} At each time $t$, characterize the states $\mathbf{x}_t$ compatible with the data $\mathbf{y}_{1:t}$ and the zonotopic mixture noise model of Definition~\ref{def:zonotopic_mixture} through a finite collection of zonotopes whose union is guaranteed to contain $\mathbf{x}_t$. This collection should carry probabilistic information about the mode realizations so that state coverage probabilities can be computed, i.e., how likely $\mathbf{x}_t$ is to lie in a given subcollection of zonotopes.
 
\item \label{prob:tractability} Develop a zonotopic mixture reduction scheme that bounds the number of retained zonotopes as $t$ increases, while preserving the state coverage probabilistic guarantees of \ref{prob:enclosure}.
\end{enumerate}

\vspace{-0.2cm}
\section{The zonotopic Kalman filter}
\label{sec:zkf}
\vspace{-0.2cm}
The critical difficulty surrounding the filtering problem for the system in \eqref{eqn:system} is that the sequence of zonotopes that generates $\mathbf{w}_t$ and $\mathbf{v}_t$, i.e., the \textit{mode history}, is not known in advance. To formalize the concept of mode history, we collect the realized mode indices up to each filtering stage into the random vectors
\begin{align}
\label{htk}
\mathbf{h}_{t|t\hspace{-0.02cm}-\hspace{-0.02cm}1}^k \hspace{-0.02cm}&=\hspace{-0.02cm} [n_1^k,m_1^k,n_2^k,m_2^k,\dots,n_{t-1}^k,m_{t-1}^k]\in \mathcal{H}_{t|t-1},   \\
\label{htl}
\mathbf{h}_{t|t}^\ell \hspace{-0.02cm}&=\hspace{-0.02cm} [\mathbf{h}_{t|t-1}^k,n_t^j] \in \mathcal{H}_{t|t},
\end{align}
so that $\mathbf{h}_{t|t-1}^k$ records the output and process modes driving the recursion up to the prediction at time $t$, and $\mathbf{h}_{t|t}^\ell$ appends the output mode $n_t^j$, $j\in\{1,\dots,V\}$, active at the measurement update, with $\ell=(k-1)V+j$. Their numbers of distinct realizations are $|\mathcal{H}_{t|t-1}|=(VW)^{t-1}$ and $|\mathcal{H}_{t|t}|=V(VW)^{t-1}$, indexed by $k=1,2,\dots,|\mathcal{H}_{t|t-1}|$ and $\ell=1,2,\dots,|\mathcal{H}_{t|t}|$, respectively. The initial history $\mathbf{h}_{1|0}$ is defined as the empty tuple.

If the realization of the mode history $\mathbf{h}_{t|t}^\ell$ were known in advance for all $t$, then the zonotopic Kalman filter (ZKF) \cite{combastel2015zonotopes} provides the zonotopes for the prediction and measurement-updated state vector. The ZKF recursion, which serves as the building block of the filter proposed in Section~\ref{sec:zmf}, is presented in Lemma~\ref{lem:zkf}.

\begin{lemma}[Zonotopic Kalman filter, \cite{combastel2015zonotopes}]
\label{lem:zkf}
Consider the system \eqref{eqn:system} with fixed and known mode sequences $\{m_\tau\}_{\tau\geq 1}$
and $\{n_\tau\}_{\tau\geq 1}$ describing the process and output-noise zonotopes at time $t$ given by
$\langle \bm{\mu}_{m_t}, \mathbf{Q}_{m_t} \rangle$ and $\langle \bm{\rho}_{n_t}, \mathbf{R}_{n_t} \rangle$, respectively, and assume that $\mathbf{x}_1\in\langle \bar{\mathbf{x}}_1, \mathbf{E}_1 \rangle$. Given the measurements $\mathbf{y}_{1:t}$ and the zonotopic bound $\mathbf{x}_t \in \langle \hat{\mathbf{x}}_{t|t-1}, \mathbf{E}_{t|t-1} \rangle$, the measurement-updated state satisfies $\mathbf{x}_t \in \langle \hat{\mathbf{x}}_{t|t}, \mathbf{E}_{t|t} \rangle$, where
\begin{align}
    \mathbf{L}_t &= \mathbf{P}_{t|t-1}\mathbf{C}^\top \left( \mathbf{C}\mathbf{P}_{t|t-1}\mathbf{C}^\top + \mathbf{R}_{n_t}\mathbf{R}_{n_t}^\top \right)^{-1}, \label{eqn:zkf_gain}\\
    \hat{\mathbf{x}}_{t|t} &= \hat{\mathbf{x}}_{t|t-1} + \mathbf{L}_t\left( \mathbf{y}_t - \mathbf{C}\hat{\mathbf{x}}_{t|t-1} - \mathbf{D}\mathbf{u}_t - \bm{\rho}_{n_t} \right), \label{eqn:zkf_center}\\
    \mathbf{E}_{t|t} &= \left[ \left( \mathbf{I} - \mathbf{L}_t\mathbf{C} \right)\mathbf{E}_{t|t-1}, \hspace{0.04cm} \mathbf{L}_t\mathbf{R}_{n_t} \right], \label{eqn:zkf_gen}
\end{align}
with $\mathbf{P}_{t|t-1} = \mathbf{E}_{t|t-1}\mathbf{E}_{t|t-1}^\top$. For every $\mathbf{W}\succ\mathbf{0}$, the gain $\mathbf{L}_t$ in \eqref{eqn:zkf_gain} minimizes the $F_\mathbf{W}$-radius of the updated zonotope $\langle \hat{\mathbf{x}}_{t|t}, \mathbf{E}_{t|t} \rangle$ over all gain matrices in $\mathbb{R}^{n_x \times n_y}$. The time-updated state satisfies $\mathbf{x}_{t+1} \in \langle \hat{\mathbf{x}}_{t+1|t}, \mathbf{E}_{t+1|t} \rangle$, where
\begin{align}
    \hat{\mathbf{x}}_{t+1|t} &= \mathbf{A}\hat{\mathbf{x}}_{t|t} + \mathbf{B}\mathbf{u}_t + \bm{\mu}_{m_t}, \label{eqn:zkf_pred_center}\\
    \mathbf{E}_{t+1|t} &= \left[ \mathbf{A}\mathbf{E}_{t|t} , \hspace{0.04cm} \mathbf{Q}_{m_t} \right]. \label{eqn:zkf_pred_gen}
\end{align}
\end{lemma}
\begin{rem}
    The reduction operator $\redu{q,\mathbf{W}}$ is typically needed in the computation of $\mathbf{E}_{t|t}$ and $\mathbf{E}_{t+1|t}$ to keep the representation tractable by avoiding the unbounded increase of columns in the generator matrix updates \eqref{eqn:zkf_gen} and \eqref{eqn:zkf_pred_gen}. This step yields an overapproximation of the zonotopic constraints for $\mathbf{x}_t$ and $\mathbf{x}_{t+1}$.
\end{rem}
While serving as the zonotopic analog of the Kalman filter \cite{anderson1979}, the ZKF does not address the mixed stochastic--deterministic nature of the zonotopic mixture noise introduced in Definition \ref{def:zonotopic_mixture}, as its recursion requires the active mode indices $m_t$ and $n_t$ at every time step, which are unavailable in our setting. A straightforward solution is to discard the probabilistic structure of the noise altogether and retain only its set-membership description, treating the noise as unknown but bounded within the zonotope unions
\begin{equation}
\label{eqn:union_bounds}
    \mathbf{w}_t\in \bigcup_{i=1}^W \langle \bm{\mu}_i, \mathbf{Q}_i \rangle, \quad \mathbf{v}_t\in \bigcup_{j=1}^V \langle \bm{\rho}_j, \mathbf{R}_j \rangle.
\end{equation}
The ZKF becomes directly implementable in this setting once each union is overapproximated by a single zonotope, e.g., by repeated application of the convex hull operation \cite{girard2005reachability}. However, this approach has two shortcomings. First, zonotope unions are in general nonconvex, so a single-zonotope enclosure is conservative, with conservatism increasing with the separation between the mode centers. Second, it discards the mode probabilities $\alpha_i$ and $\eta_j$, so no enclosure with an attached confidence level can be derived. Both shortcomings are addressed in the next section by introducing a state filter tailored to zonotopic mixture noise that pairs set-membership state bounds with mode-history probabilities.

\vspace{-0.2cm}
\section{The zonotopic mixture filter}
\label{sec:zmf}
\vspace{-0.2cm}
The proposed zonotopic mixture filter (ZMF) assigns a ZKF to each mode history, discards the histories falsified by the data, and weights the remaining ones by their posterior probability of realization. Since for a fixed mode history the noise is UBB and carries no probability distribution a priori, the data $\mathbf{y}_{1:t}$ do not admit a likelihood given a mode history. Rather, the only information the data certifies about a history is whether it is consistent with the model, i.e. whether some admissible noise realization reproduces the measurements. The following subsections formalize these notions and state the ZMF.

\subsection{Mode history probabilities and data consistency}
Since the mode sequences are i.i.d. and mutually independent, the prior probability mass function of each history is known in closed form:
\begin{equation}
\pi(\mathbf{h}^k_{t|t\hspace{-0.02cm}-\hspace{-0.02cm}1})
\hspace{-0.05cm}:= \mathbb{P}\{\mathbf{h}_{t|t-1} \hspace{-0.05cm}=\hspace{-0.05cm} \mathbf{h}^k_{t|t-1}\}
\hspace{-0.05cm}=\hspace{-0.05cm} \prod_{\tau=1}^{t-1}\eta_{n_\tau^k}\prod_{s=1}^{t-1}\alpha_{m_s^k},
\label{eq:prior}
\end{equation}
and analogously for $\pi(\mathbf{h}^\ell_{t|t})$. Not every mode history and state trajectory are compatible with the collected data. These notions are made precise next.

\begin{definition}
A history $\mathbf{h}^k_{t|t\hspace{-0.02cm}-\hspace{-0.02cm}1}$ is \emph{consistent} with the data $\mathbf{y}_{1:t-1}$ if there exist $\mathbf{x}_1 \in \langle\bar{\mathbf{x}}_1,\mathbf{E}_1\rangle$, $\mathbf{w}_\tau \in
\langle\boldsymbol{\mu}_{m_\tau^k},\mathbf{Q}_{m_\tau^k}\rangle$ and $\mathbf{v}_\tau \in\langle \boldsymbol{\rho}_{n_\tau^k},\mathbf{R}_{n_\tau^k}\rangle$ for $\tau = 1,\dots,t-1$ that generate $\mathbf{y}_{1:t-1}$ through \eqref{eqn:system}; otherwise, the history is \emph{falsified}. 

Similarly, a state trajectory is consistent with $\mathbf{h}^\ell_{t|t}$ and $\mathbf{y}_{1:t}$ if there exist $\mathbf{x}_1 \in \langle\bar{\mathbf{x}}_1,\mathbf{E}_1\rangle$, $\mathbf{w}_\tau \in\langle\boldsymbol{\mu}_{m_\tau^k},\mathbf{Q}_{m_\tau^k}\rangle$, $\mathbf{v}_\tau \in\langle \boldsymbol{\rho}_{n_\tau^k},\mathbf{R}_{n_\tau^k}\rangle$ for $\tau = 1,\dots,t-1$ and $\mathbf{v}_t \in\langle \boldsymbol{\rho}_{n_t^j},\mathbf{R}_{n_t^j}\rangle$ that generate $\mathbf{y}_{1:t}$ through \eqref{eqn:system}. 
\end{definition}

For each mode history $\mathbf{h}_{t|t}^\ell$, each generating a state containment condition $\mathbf{x}_{t}\in \langle \hat{\mathbf{x}}_{t|t-1}^k,\mathbf{E}_{t|t-1}^k \rangle$ from $\mathbf{y}_{1:t-1}$, consistency is tested recursively through the innovation indicators
\begin{equation}
c^\ell_t := \mathbbm{1}\big\{ \bm{\nu}_t^\ell\in
\zon{\mathbf{0}}{[\mathbf{C}\mathbf{E}^k_{t|t-1},\,\mathbf{R}_j]}\big\},
\label{eq:ctest}
\end{equation}
where $\bm{\nu}_t^\ell = \mathbf{y}_t - \mathbf{C}\hat{\mathbf{x}}^k_{t|t-1}
- \mathbf{D}\mathbf{u}_t - \boldsymbol{\rho}_j$. These innovation indicators accumulate along each mode history as
\begin{equation}
b^\ell_{t|t} = b^k_{t|t-1}\, c^\ell_t, \qquad
b^k_{t+1|t} = b^\ell_{t|t}, \qquad b^1_{1|0} = 1,
\label{eq:brec}
\end{equation}
defining the set of surviving histories $\mathcal{A}_{t|t} := \{\mathbf{h}^\ell_{t|t} \colon b^\ell_{t|t} = 1\}$, and analogously $\mathcal{A}_{t+1|t}$.

\subsection{Zonotopic mixture filter}

Theorem \ref{thm:zmf} introduces the zonotopic mixture filter and its main properties.
\begin{theorem}[Zonotopic mixture filter]
\label{thm:zmf}
Consider the system \eqref{eqn:system} under the zonotopic mixture noise model of Section \ref{subsec:system}, with data $\{\mathbf{u}_{1:N},\mathbf{y}_{1:N}\}$. Initialize $M_{1|0} = 1$, $\gamma^1_{1|0} = 1$, $b^1_{1|0} = 1$, $\hat{\mathbf{x}}^1_{1|0} = \bar{\mathbf{x}}_1$, and $\mathbf{E}^1_{1|0} = \mathbf{E}_1$, and compute recursively, for $t = 1,\dots,N$:

\emph{Measurement update}: For $\ell = (k-1)V + j$ with
$k = 1,\dots,M_{t|t-1}$ and $j \in \mathcal{I}_V$, set
$M_{t|t} = VM_{t|t-1}$, evaluate $b^\ell_{t|t}$ via \eqref{eq:brec}, and
\begin{align}
\gamma^\ell_{t|t} &= \bar{\gamma}^\ell_{t|t}
\Big(\textstyle\sum_{s=1}^{M_{t|t}}\bar{\gamma}^s_{t|t}\Big)^{-1}, \\
\bar{\gamma}^\ell_{t|t} &= \eta_j\,\gamma^k_{t|t-1}\mathbbm{1}\big\{ \bm{\nu}_t^\ell\in
\zon{\mathbf{0}}{[\mathbf{C}\mathbf{E}^k_{t|t-1},\,\mathbf{R}_j]}\big\},
\label{eq:wmeas}\\
\bm{\nu}_t^\ell &=  \mathbf{y}_t - \mathbf{C}\hat{\mathbf{x}}^k_{t|t-1}
- \mathbf{D}\mathbf{u}_t - \boldsymbol{\rho}_j, \\
\label{eq:xmeas}
\hat{\mathbf{x}}^\ell_{t|t} &= \hat{\mathbf{x}}^k_{t|t-1}
+ \mathbf{G}^\ell_t \bm{\nu}_t^\ell,\\
\mathbf{E}^\ell_{t|t} &= \big[(\mathbf{I}
- \mathbf{G}^\ell_t\mathbf{C})\mathbf{E}^k_{t|t-1},\;
\mathbf{G}^\ell_t\mathbf{R}_j\big],
\end{align}
where $\mathbf{G}_t^\ell\in\mathbb{R}^{n_x\times n_y}$ are arbitrary gain matrices.

\emph{Time update}: For $k = (\ell-1)W + i$ with $\ell = 1,\dots,M_{t|t}$
and $i \in \mathcal{I}_W$, set $M_{t+1|t} = WM_{t|t}$,
$b^k_{t+1|t} = b^\ell_{t|t}$, and
\begin{align}
\gamma^k_{t+1|t} &= \alpha_i\,\gamma^\ell_{t|t},\\
\hat{\mathbf{x}}^k_{t+1|t} &= \mathbf{A}\hat{\mathbf{x}}^\ell_{t|t}
+ \mathbf{B}\mathbf{u}_t + \boldsymbol{\mu}_i,\\
\label{generator_timeupdate}
\mathbf{E}^k_{t+1|t} &= \big[\mathbf{A}\mathbf{E}^\ell_{t|t},\;
\mathbf{Q}_i\big].
\end{align}
Then, for every $t = 1,\dots,N$, the following statements hold:
\begin{enumerate}
\item[(i)] For every $\ell=1,\dots,M_{t|t}$, every state trajectory consistent with $\mathbf{h}^\ell_{t|t}$ and $\mathbf{y}_{1:t}$ satisfies $\mathbf{x}_t\in\langle\hat{\mathbf{x}}^\ell_{t|t}, \mathbf{E}^\ell_{t|t}\rangle$ and, whenever such trajectory exists, $\mathbf{h}_{t|t}^\ell\in \mathcal{A}_{t|t}$. Conversely, $b^\ell_{t|t} = 0$ implies that the mode history $\mathbf{h}^\ell_{t|t}$ is falsified by the data.
\item[(ii)] The probability mass function of the mode history $\mathbf{h}_{t|t}$, evaluated at $\mathbf{h}^\ell_{t|t}$, conditioned on the consistency event $\{\mathbf{h}_{t|t} \in \mathcal{A}_{t|t}\}$ is described by the weights
\begin{equation}
\gamma^\ell_{t|t}
= \frac{\pi\big(\mathbf{h}^\ell_{t|t}\big)\,b^\ell_{t|t}}
{\sum_{s=1}^{M_{t|t}}\pi\big(\mathbf{h}^s_{t|t}\big)\,b^s_{t|t}}.
\label{eq:wexact}
\end{equation}
\item[(iii)] For every $\mathbf{W}\succ \mathbf{0}$, the following choice of $\mathbf{G}_t^\ell$ minimizes the $F_\mathbf{W}$-radius of the updated zonotope $\langle\hat{\mathbf{x}}^\ell_{t|t}, \mathbf{E}^\ell_{t|t}\rangle$ over all gain matrices in $\mathbb{R}^{n_x\times n_y}$:
\begin{equation}
\label{optimalgainzmf}
\mathbf{G}^\ell_t\hspace{-0.04cm}=\hspace{-0.04cm}\mathbf{L}^\ell_t \hspace{-0.04cm}:= \hspace{-0.04cm} \mathbf{P}_{t|t-1}^k \mathbf{C}^\top\big(\mathbf{C}\mathbf{P}_{t|t-1}^k\mathbf{C}^\top
+\mathbf{R}_j\mathbf{R}_j^\top\big)^{-1}, 
\end{equation}
where $\mathbf{P}_{t|t-1}^k:=\mathbf{E}^k_{t|t-1}\mathbf{E}^{k\top}_{t|t-1}$.
 
\end{enumerate}
The analogous claims of Statements (i) and (ii) hold for the predicted quantities
$\{(\gamma^k_{t+1|t},
\zon{\hat{\mathbf{x}}^k_{t+1|t}}{\mathbf{E}^k_{t+1|t}})\}_{k=1}^{M_{t+1|t}}$
with respect to $\mathbf{y}_{1:t}$ and $\mathcal{A}_{t+1|t}$.
\end{theorem}

\begin{pf}
\emph{Part (i)}. We proceed by induction on the filter steps. At $t = 1$, prior to the measurement update, every state trajectory satisfies $\mathbf{x}_1 \in \zon{\bar{\mathbf{x}}_1}{\mathbf{E}_1} =
\zon{\hat{\mathbf{x}}^1_{1|0}}{\mathbf{E}^1_{1|0}}$ by assumption, and $b^1_{1|0} = 1$.

Fix $\ell = (k-1)V+j$ and suppose that every trajectory consistent with $\mathbf{h}^k_{t|t-1}$ and
$\mathbf{y}_{1:t-1}$ satisfies $\mathbf{x}_t \in \zon{\hat{\mathbf{x}}^k_{t|t-1}}{\mathbf{E}^k_{t|t-1}}$ and $b^k_{t|t-1} = 1$. Consider any trajectory consistent with $\mathbf{h}^\ell_{t|t}$ and $\mathbf{y}_{1:t}$. Then, $\mathbf{x}_t \in \zon{\hat{\mathbf{x}}^k_{t|t-1}}{\mathbf{E}^k_{t|t-1}}$, and together with $\mathbf{v}_t\in \langle \bm{\rho}_j,\mathbf{R}_j \rangle$ this implies
\begin{align}
    \bm{\nu}_t^\ell = \mathbf{C}(\mathbf{x}_t-\hat{\mathbf{x}}_{t|t-1}^k)+\mathbf{v}_t-\bm{\rho}_j \in \langle \mathbf{0},\mathbf{E}_{t|t-1}^k \rangle \oplus \langle \mathbf{0},\mathbf{R}_{j} \rangle, \notag
\end{align}
hence $c^\ell_t = 1$. By \eqref{eq:brec}, $b^\ell_{t|t} = 1$, i.e., $\mathbf{h}^\ell_{t|t} \in \mathcal{A}_{t|t}$. Furthermore, replacing $\mathbf{y}_t=\mathbf{Cx}_t+\mathbf{Du}_t+\mathbf{v}_t$ in \eqref{eq:xmeas} yields
\begin{align}
    \mathbf{x}_t-\hat{\mathbf{x}}^\ell_{t|t} &= (\mathbf{I}-\mathbf{G}_t^\ell \mathbf{C})(\mathbf{x}_t-\hat{\mathbf{x}}^k_{t|t-1})-\mathbf{G}_t^\ell (\mathbf{v}_t-\bm{\rho}_j) \notag \\
    \implies \mathbf{x}_t-\hat{\mathbf{x}}^\ell_{t|t} &\in \langle \mathbf{0}, (\mathbf{I}-\mathbf{G}_t^\ell \mathbf{C})\mathbf{E}^k_{t|t-1} \rangle \oplus \zon{\mathbf{0}}{\mathbf{G}_t^\ell\mathbf{R}_j}, \notag 
\end{align}
and therefore $\mathbf{x}_t \in\zon{\hat{\mathbf{x}}^\ell_{t|t}}{\mathbf{E}^\ell_{t|t}}$.

For the time update, we note that any trajectory consistent with
$\mathbf{h}^k_{t+1|t} = [\mathbf{h}^\ell_{t|t}, i]$ and $\mathbf{y}_{1:t}$
satisfies $\mathbf{x}_{t+1} = \mathbf{A}\mathbf{x}_t +
\mathbf{B}\mathbf{u}_t + \boldsymbol{\mu}_i + \mathbf{Q}_i\mathbf{s}_i$
with $\mathbf{x}_t \in
\zon{\hat{\mathbf{x}}^\ell_{t|t}}{\mathbf{E}^\ell_{t|t}}$ and
$\|\mathbf{s}_i\|_\infty \leq 1$, which gives $\mathbf{x}_{t+1} \in
\zon{\hat{\mathbf{x}}^k_{t+1|t}}{\mathbf{E}^k_{t+1|t}}$. The induction is complete, and the falsification claim follows by contraposition: if $b_{t|t}^\ell=0$, the recursion admits no consistent trajectory, hence no admissible noise realization can generate $\mathbf{y}_{1:t}$, i.e., $\mathbf{h}^\ell_{t|t}$ is falsified.

\emph{Part (ii).} We prove \eqref{eq:wexact} by induction. Define $Z_{t|t} = \sum_{s=1}^{M_{t|t}} \pi(\mathbf{h}^s_{t|t})b^s_{t|t}$, and analogously for the predicted weights. The base case holds since $\gamma^1_{1|0} = \pi(\mathbf{h}^1_{1|0})b^1_{1|0} = 1$, with the prior of the empty tuple being the empty product. For the measurement step, assume $\gamma^k_{t|t-1} = \pi(\mathbf{h}^k_{t|t-1})b^k_{t|t-1}/Z_{t|t-1}$. Since $\pi(\mathbf{h}^\ell_{t|t}) = \eta_j\,\pi(\mathbf{h}^k_{t|t-1})$ by
\eqref{eq:prior} and $b^\ell_{t|t} = b^k_{t|t-1}c^\ell_t$ by~\eqref{eq:brec},
\begin{equation*}
\bar{\gamma}^\ell_{t|t} = \eta_j\,\gamma^k_{t|t-1}\,c^\ell_t = \pi\big(\mathbf{h}^\ell_{t|t}\big)\,b^\ell_{t|t}\,/\,Z_{t|t-1},
\end{equation*}
and the normalization in \eqref{eq:wmeas} replaces $Z_{t|t-1}$ by $Z_{t|t}$. For the time update, $\pi(\mathbf{h}^k_{t+1|t}) =\alpha_i\,\pi(\mathbf{h}^\ell_{t|t})$ and $b^k_{t+1|t} = b^\ell_{t|t}$, so
$\gamma^k_{t+1|t} = \alpha_i\gamma^\ell_{t|t} =
\pi(\mathbf{h}^k_{t+1|t})b^k_{t+1|t}/Z_{t|t}$, and the weights remain
normalized because $\sum_{i\in\mathcal{I}_W}\alpha_i = 1$, i.e.,
$Z_{t+1|t} = Z_{t|t}$. Finally, since
$\mathbb{P}\{\mathbf{h}_{t|t} \in \mathcal{A}_{t|t}\} =
\sum_{s:\, b^s_{t|t}=1}\pi(\mathbf{h}^s_{t|t}) = Z_{t|t}$, the ratio in
\eqref{eq:wexact} is precisely the conditional probability mass function
of $\mathbf{h}_{t|t}$ given $\{\mathbf{h}_{t|t} \in \mathcal{A}_{t|t}\}$.

\emph{Part (iii).} Direct from the proof of the ZKF in Lemma \ref{lem:zkf}, found in \cite{combastel2015zonotopes}, when applied to a fixed mode history. \hfill $\square$
\end{pf}

\begin{rem}
\label{rem:conservatism}
At $t = 1$ the test \eqref{eq:ctest} is exact, since the predicted set $\zon{\bar{\mathbf{x}}_1}{\mathbf{E}_1}$ coincides with the exact set of possible initial states. For $t \geq 2$, the measurement-updated zonotope $\mathbf{E}^\ell_{t|t}$ is an outer bound of the exact set of states consistent with $\mathbf{h}^\ell_{t|t}$ and $\mathbf{y}_{1:t}$, which is in line with the conservatism of the ZKF \cite{combastel2015zonotopes}. So, a failed test \eqref{eq:ctest} conclusively falsifies a history, while a passed test does not certify consistency. The same observation applies when the reduction operator $\downarrow_{q,\mathbf{W}}$ is employed, since $\zon{\mathbf{c}}{\mathbf{E}} \subseteq\zon{\mathbf{c}}{\downarrow_{q,\mathbf{W}}\mathbf{E}}$ preserves Statement (i) of Theorem~\ref{thm:zmf} and the soundness of falsification.
\end{rem}

Under the zonotopic mixture noise model of Definition \ref{def:zonotopic_mixture}, the cumulative data $\mathbf{y}_{1:t}$ certifies only whether a mode history is falsified, and \eqref{eq:wexact} provides the conditional distribution on the mode space, where the probabilistic structure of the model resides. The weights $\gamma_{t|t}^\ell$ and their support are invariant to the distribution of the noise within each zonotope, since no such distribution is assigned by the model. Moreover, every falsified mode history receives zero posterior probability under any noise distribution supported on the zonotopes, and thus can be discarded in the subsequent iterations.

One advantage of the ZMF over the ZKF is the possibility of determining probabilities over mode histories, which translate into state constraint probabilities over time. In particular, a selection of a strict subset $\mathcal{S}$ of high-weight mode histories can lead to a smaller zonotope enclosure for $\mathbf{x}_t$ that holds with a guaranteed probability $\sum_{\ell\in\mathcal{S}} \gamma_{t|t}^\ell$, at the expense of discarding the lowest-weight mode histories. The formalization of this result is presented in Corollary~\ref{cor:prob}.

\begin{corollary} 
\label{cor:prob}
Consider the system in \eqref{eqn:system} under the zonotopic mixture noise model of Section \ref{subsec:system}, with a fixed data record $\{\mathbf{u}_{1:N},\mathbf{y}_{1:N}\}$, from which the quantities $\{(\gamma^\ell_{t|t}, \zon{\hat{\mathbf{x}}^\ell_{t|t}} {\mathbf{E}^\ell_{t|t}})\}_{\ell=1}^{M_{t|t}}$ and $\mathcal{A}_{t|t}$ of Theorem~\ref{thm:zmf} are computed. Let $\mathcal{X}_t(\mathbf{h}_t)$ denote the set of states $\mathbf{x}_t$ consistent with the mode history $\mathbf{h}_t$ and $\mathbf{y}_{1:t}$. Then, for every $t=1,\dots,N$ and every index set $\mathcal{S}\subseteq\{1,\dots,M_{t|t}\}$,
\begin{equation}
\label{eqn:corprob}
\mathbb{P}\left\{\mathcal{X}_t(\mathbf{h}_{t|t})\subseteq \bigcup_{\ell\in \mathcal{S}} \zon{\hat{\mathbf{x}}^\ell_{t|t}}{\mathbf{E}^\ell_{t|t}}\hspace{0.03cm} \big| \, \mathbf{h}_{t|t}\in\mathcal{A}_{t|t}\right\} \geq \sum_{\ell\in\mathcal{S}} \gamma_{t|t}^\ell.
\end{equation}
\end{corollary}

\begin{pf}
Let $\mathcal{E}_t$ denote the event in \eqref{eqn:corprob}, i.e., that every $\mathbf{x}_t$ consistent with $\mathbf{h}_{t|t}$ and $\mathbf{y}_{1:t}$ belongs to $\bigcup_{s\in\mathcal{S}} \zon{\hat{\mathbf{x}}^s_{t|t}}{\mathbf{E}^s_{t|t}}$. Conditioned on $\mathbf{h}_{t|t}=\mathbf{h}^\ell_{t|t}$ with $\mathbf{h}^\ell_{t|t}\in\mathcal{A}_{t|t}$, Part (i) of Theorem \ref{thm:zmf} gives $\mathbf{x}_t\in\zon{\hat{\mathbf{x}}^\ell_{t|t}}{\mathbf{E}^\ell_{t|t}}$ for every consistent trajectory, so for $\ell\in\mathcal{S}$ the event $\mathcal{E}_t$ occurs with conditional probability one. By the law of total probability and
Part (ii) of Theorem \ref{thm:zmf},
    \begin{align}
        \mathbb{P}\left\{\mathcal{E}_t \big|\,\mathbf{h}_{t|t}\hspace{-0.03cm}\in\hspace{-0.03cm}\mathcal{A}_{t|t}\right\}\hspace{-0.06cm}&=\hspace{-0.05cm}\sum_{\ell=1}^{M_{t|t}} \hspace{-0.05cm}\gamma_{t|t}^\ell\,
\mathbb{P}\hspace{-0.03cm}\left\{\mathcal{E}_t \big|\,
\mathbf{h}_{t|t}\hspace{-0.03cm}\in\hspace{-0.03cm}\mathcal{A}_{t|t},
\mathbf{h}_{t|t}\hspace{-0.03cm}=\hspace{-0.03cm}\mathbf{h}_{t|t}^\ell\right\} \notag \\
&\geq \sum_{\ell\in\mathcal{S}} \hspace{-0.03cm}\gamma_{t|t}^\ell\,
\mathbb{P}\left\{\mathcal{E}_t \big|\,
\mathbf{h}_{t|t}\hspace{-0.03cm}\in\hspace{-0.03cm}\mathcal{A}_{t|t},
\mathbf{h}_{t|t}\hspace{-0.03cm}=\hspace{-0.03cm}\mathbf{h}_{t|t}^\ell\right\} \notag \\
&=\sum_{\ell\in\mathcal{S}} \gamma_{t|t}^\ell, \notag 
    \end{align}
where the final equality uses that $\mathcal{E}_t$ has conditional probability one on each surviving branch in $\mathcal{S}$. \hfill $\square$
\end{pf}

\begin{rem}
\label{rem:invalidation}
Under the model assumptions, the realized mode history is consistent with the data, so the normalization in \eqref{eq:wmeas} is well defined for all $t$. Conversely,  $\sum_{s=1}^{M_{t|t}} \bar{\gamma}^s_{t|t} = 0$ certifies that no admissible mode history can explain the data, which can be used for fault detection~\cite{wang2018zonotopic}.
\end{rem}

\subsection{Relation to the Gaussian sum filter and a discrete Bayes filter}
\label{sec:relations}

The ZMF is closely related in structure to two mixture filters: the Gaussian sum filter (GSF) \cite{alspach1972nonlinear,kitagawa1987non}, and the exact discrete Bayes filter (DBF), which is a type of point-mass filter \cite{bergman1999recursive}. All three share the categorical mode-selection structure of Section~\ref{sec:setup}, i.e., the indices $m_t$ and $n_t$ drawn with weights $\alpha_i$ and $\eta_j$ respectively, and differ only in how the noise is modeled within each mode. The GSF and DBF replace the zonotopic mixture noise of the ZMF by the respective stochastic models
\begin{align}
\label{gmmnoise}
    \mathbf{w}_t &\hspace{-0.03cm}\sim\hspace{-0.03cm} \sum_{i=1}^W \alpha_i \mathcal{N}(\mathbf{w}_t;\bm{\mu}_i,\bm{\Sigma}_i^w),\hspace{-0.03cm}\quad \hspace{-0.1cm}  \mathbf{v}_t \hspace{-0.03cm}\sim \hspace{-0.03cm} \sum_{j=1}^V \eta_j \mathcal{N}(\mathbf{v}_t;\bm{\rho}_j,\bm{\Sigma}_j^v), \\
    \label{pmfnoise}
    \mathbf{w}_t &\sim \sum_{i=1}^W \alpha_i \delta(\mathbf{w}_t,\bm{\mu}_i), \quad \mathbf{v}_t \sim \sum_{j=1}^V \eta_j \delta(\mathbf{v}_t,\bm{\rho}_j).
\end{align}
Under the discrete model \eqref{pmfnoise}, the reachable state set has finite cardinality at every time step, and the exact discrete Bayes filter can be computed from the point-mass recursion stated in Lemma~\ref{lem:pmf}.

\begin{lemma}
    \label{lem:pmf}
    Consider the system in \eqref{eqn:system} under the discrete noise model \eqref{pmfnoise}, with data $\{\mathbf{u}_{1:N},\mathbf{y}_{1:N}\}$. Initialize $M_{1|0} = 1$, $\gamma^1_{1|0} = 1$, and $\hat{\mathbf{x}}^1_{1|0} = \bar{\mathbf{x}}_1$, and compute recursively, for $t = 1,\dots,N$:
    
    \emph{Measurement update}: For $\ell = (k-1)V + j$ with $k = 1,\dots,M_{t|t-1}$ and $j \in \mathcal{I}_V$, set $M_{t|t} = VM_{t|t-1}$, and
    \begin{align}
    \gamma^\ell_{t|t} &= \bar{\gamma}^\ell_{t|t} \Big(\textstyle\sum_{s=1}^{M_{t|t}}\bar{\gamma}^s_{t|t}\Big)^{-1},\\
    \bar{\gamma}^\ell_{t|t} &= \eta_j\,\gamma^k_{t|t-1} \delta\big(\bm{\nu}_t^\ell,\mathbf{0}\big), \label{eqn:bar_gamma_filtering_theorem}\\
    \bm{\nu}_t^\ell &= \mathbf{y}_t - \mathbf{C}\hat{\mathbf{x}}^k_{t|t-1}-\mathbf{D}\mathbf{u}_t - \boldsymbol{\rho}_j,\\
    \label{eqn:hatxfiltering}
    \hat{\mathbf{x}}^\ell_{t|t} &= \hat{\mathbf{x}}^k_{t|t-1}.
    \end{align}

    \emph{Time update}: For $k = (\ell-1)W + i$ with $\ell = 1,\dots,M_{t|t}$ and $i \in \mathcal{I}_W$, set $M_{t+1|t} = WM_{t|t}$, and
    \begin{align}
    \gamma^k_{t+1|t} &= \alpha_i\,\gamma^\ell_{t|t},
    \label{eqn:bargamma1_filtering_theorem}\\
    \hat{\mathbf{x}}^k_{t+1|t} &= \mathbf{A}\hat{\mathbf{x}}^\ell_{t|t} + \mathbf{B}\mathbf{u}_t + \boldsymbol{\mu}_i.
    \label{eqn:hat_x1_filtering_theorem}
    \end{align}
    Then, for every $t = 1,\dots,N$, the reachable state set $\mathcal{X}_t$ is finite, and the exact Bayesian posteriors are given by
    \begin{align}
    \label{eqn:filtering_pmf}
    P(\mathbf{x}_t|\mathbf{y}_{1:t}) &= \sum_{\ell=1}^{M_{t|t}} \gamma^\ell_{t|t} \delta\big(\mathbf{x}_t,\hat{\mathbf{x}}^\ell_{t|t}\big),\\
    \label{pmf_prediction}
    P(\mathbf{x}_{t+1}|\mathbf{y}_{1:t}) &=\sum_{k=1}^{M_{t+1|t}} \gamma^k_{t+1|t} \delta\big(\mathbf{x}_{t+1},\hat{\mathbf{x}}^k_{t+1|t}\big).
    \end{align}
\end{lemma}
\begin{pf}
    See Appendix A. \hfill$\square$
\end{pf}

Because of their shared categorical structure, the three filters admit a common interpretation as a bank of base filters running in parallel, one per mode history, whose weights are updated by the agreement between each branch's predicted output and the measurement. The branches are identical across the three filters, with the GSF running a Kalman filter per history, the ZMF a ZKF, and the DBF a deterministic propagation. They share the same time update $\hat{\mathbf{x}}^k_{t+1|t}$ and the same branch bookkeeping $M_{t|t}=VM_{t|t-1}$, $M_{t+1|t}=WM_{t|t}$. The filter recursions differ only in the weight updates, given by
\begin{align}
    \hspace{-0.4cm}\textnormal{(ZMF)}\,\,\,\bar{\gamma}^\ell_{t|t} &= \eta_j\gamma^k_{t|t-1}\mathbbm{1}\big\{\boldsymbol{\nu}^\ell_t \in
\zon{\mathbf{0}}{[\mathbf{C}\mathbf{E}^k_{t|t-1},\mathbf{R}_j]}\big\}, \label{eqn:weight_zmf}\\
    \hspace{-0.4cm}\textnormal{(GSF)}\,\,\,\bar{\gamma}^\ell_{t|t} &= \eta_j\gamma^k_{t|t-1} \mathcal{N}(\boldsymbol{\nu}^\ell_t;\mathbf{0}, \mathbf{S}^\ell_t), \label{eqn:weight_gsf}\\
    \hspace{-0.4cm}\textnormal{(DBF)}\,\,\,\bar{\gamma}^\ell_{t|t} &= \eta_j \gamma_{t|t-1}^{k}\delta(\boldsymbol{\nu}^\ell_t,\mathbf{0}), \label{eqn:weight_dbf}
\end{align}
where $\mathbf{S}^\ell_t$ denotes the innovation covariance for the GSF setting. Notably, the covariation matrix of the innovation in the ZMF is
\begin{equation}
    \big[\mathbf{C}\mathbf{E}^k_{t|t-1},\,\mathbf{R}_j \big]
    \big[\mathbf{C}\mathbf{E}^k_{t|t-1},\,\mathbf{R}_j \big]^\top = \mathbf{C}\mathbf{P}^k_{t|t-1}\mathbf{C}^\top + \mathbf{R}_j \mathbf{R}_j^\top, \notag
\end{equation}
which corresponds in form to the covariance $\mathbf{S}^\ell_t$ computed in the GSF by replacing the covariation $\mathbf{R}_j \mathbf{R}_j^\top$ with the Gaussian component covariance $\bm{\Sigma}_j^v$.

To formalize the relation between the three weight rules, we study the ZMF and the GSF in the small-dispersion regime in which the noise generated from each mode concentrates around the centers $\bm{\mu}_i$ and $\bm{\rho}_j$. Theorem~\ref{thm:filters-dbf} shows that, when the data is generated by the discrete model \eqref{pmfnoise}, both filters reduce to the DBF, although in different senses. The ZMF coincides with the DBF for every dispersion below a positive threshold, whereas the GSF converges to the DBF only as the dispersion tends to zero, and under the assumption that the component covariances do not depend on the mode indices.

\begin{theorem}
\label{thm:filters-dbf}
Consider the system \eqref{eqn:system}, where the data
$\mathbf{y}_{1:N}$ is generated by the discrete noise model
\eqref{pmfnoise} with deterministic initial state $\bar{\mathbf{x}}_1$, and let
$\varepsilon>0$ be a dispersion parameter. The following statements hold:
\begin{enumerate}
\item[(i)] If the ZMF is implemented with optimal gain \eqref{optimalgainzmf}, noise zonotopes $\zon{\bm{\mu}_i}{\varepsilon\mathbf{Q}_i}$ and $\zon{\bm{\rho}_j}{\varepsilon\mathbf{R}_j}$, and initial zonotope $\zon{\bar{\mathbf{x}}_1}{\varepsilon\mathbf{E}_1}$, then there exists $\varepsilon^*>0$ such that, for every $\varepsilon\in(0,\varepsilon^*)$ and every $t\leq N$, the ZMF weights $\gamma^{\ell}_{t|t}$ and $\gamma^{k}_{t+1|t}$ coincide with those of the DBF in Lemma~\ref{lem:pmf}. Moreover, on every branch of positive weight, the ZMF centers equal the corresponding DBF atoms, and the generator matrices scale linearly with $\varepsilon$.
\item[(ii)] If the GSF is implemented with the Gaussian mixture
\eqref{gmmnoise} with mode-independent covariances
$\bm{\Sigma}_i^w=\varepsilon^2\bm{\Sigma}^w$ and
$\bm{\Sigma}_j^v=\varepsilon^2\bm{\Sigma}^v$, where
$\bm{\Sigma}^w,\bm{\Sigma}^v\succ\mathbf{0}$, and initial density
$\mathcal{N}(\mathbf{x}_1;\bar{\mathbf{x}}_1,\varepsilon^2\bm{\Sigma}_1)$,
then, as $\varepsilon\to 0$, the GSF weights $\gamma^{\ell}_{t|t}$ and
$\gamma^{k}_{t+1|t}$ converge to the respective DBF weights for every $t\leq N$, and the GSF posterior converges weakly to the DBF posterior \eqref{eqn:filtering_pmf}.
\end{enumerate}
\end{theorem}
\vspace{-0.25cm}
\begin{pf}
    See Appendix B. \hfill$\square$
\end{pf}
\vspace{-0.25cm}

The connections in Theorem \ref{thm:filters-dbf} are illustrated in Figure~\ref{fig:filter-relations}. The asymmetry between the convergence results stems from the support of each noise model. The bounded support of the zonotopic mixtures allows the ZMF to falsify incompatible histories for finite $\varepsilon$, whereas the full support of the Gaussian mixtures in the GSF downweights the history likelihoods without ever eliminating them, so exact agreement is reached only in the limit.

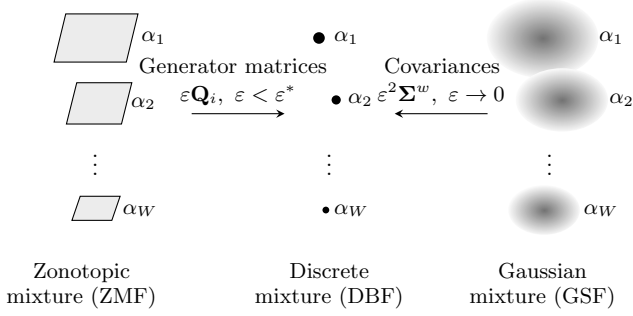
\begin{figure}[t]
\centering
\begin{tikzpicture}[>=stealth, font=\footnotesize, scale=0.91, transform shape]
  \draw[fill=gray!15] (-4.0,0.85) -- (-2.95,0.85) -- (-2.78,1.55) -- (-3.83,1.55) -- cycle;
  \draw[fill=gray!15] (-3.82,-0.05) -- (-3.02,-0.05) -- (-2.87,0.55) -- (-3.67,0.55) -- cycle;
   \draw[fill=gray!15] (-3.74,-1.45) -- (-3.16,-1.45) -- (-3.04,-1.1) -- (-3.62,-1.1) -- cycle;
   \node at (-2.55,1.2)   {$\alpha_1$};
  \node at (-2.70,0.25)  {$\alpha_2$};
  \node at (-2.75,-1.3)  {$\alpha_W$};
  \node at (-3.4,-0.55) {$\vdots$};
  \fill (-0.15,1.2)  circle (2.4pt); \node at (0.25,1.2)  {$\alpha_1$};
  \fill ( 0.10,0.3)  circle (1.9pt); \node at (0.45,0.3)  {$\alpha_2$};
  \fill (-0.05,-1.3) circle (1.4pt); \node at (0.32,-1.3) {$\alpha_W$};
  \node at (0.0,-0.55) {$\vdots$};
  \shade[inner color=black!55, outer color=black!2] (3.1,1.2) ellipse (0.81 and 0.58);
  \shade[inner color=black!55, outer color=black!2] (3.38,0.3) ellipse (0.67 and 0.46);
  \shade[inner color=black!55, outer color=black!2] (3.12,-1.3) ellipse (0.52 and 0.36);
  \node at (4.1,1.2)  {$\alpha_1$};
  \node at (4.25,0.3)  {$\alpha_2$};
  \node at (3.95,-1.3) {$\alpha_W$};
  \node at (3.2,-0.55) {$\vdots$};
  \node at (-1.4,0.8) {Generator matrices};
  \draw[->] (-2,0.1) -- (-0.65,0.1) node[midway,above] {$\varepsilon\mathbf{Q}_i,\ \varepsilon< \varepsilon^*$};
  \node at (1.64,0.8) {Covariances};
  \draw[->] ( 2.3,0.1) -- ( 0.95,0.1) node[midway,above] {$\varepsilon^2\bm{\Sigma}^w,\ \varepsilon\to 0$};
  \node[align=center] at (-3.6,-2.4) {Zonotopic\\mixture (ZMF)};
  \node[align=center] at ( 0.0,-2.4)  {Discrete\\mixture (DBF)};
  \node[align=center] at ( 3.1,-2.4)  {Gaussian\\mixture (GSF)};
\end{tikzpicture}
\vspace{-0.55cm}
\caption{Relationship between the ZMF, GSF, and DBF. The three filters share the same mixture structure over the categorical distribution with weights $\alpha_i$ (only the state noise is depicted). The ZMF coincides with the DBF for every dispersion below a finite threshold $\varepsilon^*$, while the GSF with mode-independent covariances converges to the DBF as $\varepsilon\to 0$.}
\label{fig:filter-relations}
\end{figure}

\vspace{-0.15cm}
\section{Zonotopic mixture reduction}
\label{sec:implementation}
\vspace{-0.15cm}
Analogously to the GSF \cite{alspach1972nonlinear}, the number of components describing each state zonotopic mixture in Theorem~\ref{thm:zmf} grows exponentially over time, which limits the applicability of the ZMF if left unhandled. Hence, this section proposes a zonotopic mixture reduction (ZMR) scheme that maintains computational tractability while keeping the probabilistic enclosure properties of the filter.

A simple approach is to discard the components with the lowest weights, known as pruning in the Gaussian mixture reduction literature \cite{crouse2011look}. This method is inexpensive, but the union of the retained zonotopes no longer contains every state consistent with the data, so the enclosure statements of Theorem~\ref{thm:zmf} and Corollary~\ref{cor:prob} no longer hold. We therefore first formalize the requirements that a ZMR scheme must meet for these guarantees to be preserved, and then present the proposed reduction mechanism.

\vspace{-0.15cm}
\subsection{Admissible reductions}
\vspace{-0.15cm}

In contrast to Gaussian sum reduction, which is commonly posed as a moment-matching or density-approximation problem \cite{crouse2011look}, ZMR must preserve the set-membership property of Definition~\ref{def:zonotopic_mixture}, in the sense that every realization of the original mixture must remain a realization of the reduced one. This requirement is formalized next.

\begin{definition}[Admissible ZMR]
\label{def:zm_reduction}
Consider the zonotopic mixture $\{(\alpha_i,\langle \mathbf{c}_i, \mathbf{E}_i \rangle)\}_{i=1}^{M}$, and let $\{\mathcal{I}_j\}_{j=1}^{m}$, with $m\leq M$, be a partition of $\{1,\dots,M\}$. A zonotopic mixture $\{(\tilde{\alpha}_j,\langle \tilde{\mathbf{c}}_j, \tilde{\mathbf{E}}_j \rangle)\}_{j=1}^{m}$ is an \emph{admissible reduction} of $\{(\alpha_i,\langle \mathbf{c}_i, \mathbf{E}_i \rangle)\}_{i=1}^{M}$ if, for all $j\in\{1,\dots,m\}$,
\begin{equation}
    \langle \tilde{\mathbf{c}}_j, \tilde{\mathbf{E}}_j \rangle \supseteq
    \bigcup_{i\in\mathcal{I}_j} \langle \mathbf{c}_i, \mathbf{E}_i \rangle,
    \qquad
    \tilde{\alpha}_j = \sum_{i\in\mathcal{I}_j} \alpha_i.
    \label{eqn:zm_reduction_conditions}
\end{equation}
\end{definition}
\vspace{-0.25cm}

The two conditions in \eqref{eqn:zm_reduction_conditions} can be exploited to transfer the probabilistic enclosure guarantees of Corollary~\ref{cor:prob} to the reduced mixture, as shown in Corollary \ref{cor:reduced_prob}.

\begin{corollary}
\label{cor:reduced_prob}
Consider the setup of Corollary~\ref{cor:prob}, and let $\{(\tilde{\gamma}^{j}_{t|t},\zon{\tilde{\mathbf{x}}^{j}_{t|t}}{\tilde{\mathbf{E}}^{j}_{t|t}})\}_{j=1}^{m_{t|t}}$ be an admissible reduction of $\{(\gamma^{\ell}_{t|t},\zon{\hat{\mathbf{x}}^{\ell}_{t|t}}{\mathbf{E}^{\ell}_{t|t}})\}_{\ell=1}^{M_{t|t}}$ with partition $\{\mathcal{I}_j\}_{j=1}^{m_{t|t}}$. Then, for every index set $\tilde{\mathcal{S}}\subseteq\{1,\dots,m_{t|t}\}$,
\begin{equation}
\mathbb{P}\bigg\{\mathcal{X}_t(\mathbf{h}_{t|t})\subseteq \bigcup_{j\in \tilde{\mathcal{S}}}
\zon{\tilde{\mathbf{x}}^{j}_{t|t}}{\tilde{\mathbf{E}}^{j}_{t|t}}\hspace{0.03cm} \big| \, \mathbf{h}_{t|t}\in\mathcal{A}_{t|t}\bigg\} \geq \sum_{j\in\tilde{\mathcal{S}}} \tilde{\gamma}_{t|t}^{j}.
\label{eqn:reduced_enclosure}
\end{equation}
\end{corollary}
\vspace{-0.25cm}
\begin{pf}
Define $\mathcal{S}:=\bigcup_{j\in\tilde{\mathcal{S}}}\mathcal{I}_j\subseteq\{1,\dots,M_{t|t}\}$. The first condition in \eqref{eqn:zm_reduction_conditions} implies $\bigcup_{j\in\tilde{\mathcal{S}}}\zon{\tilde{\mathbf{x}}^{j}_{t|t}}{\tilde{\mathbf{E}}^{j}_{t|t}} \supseteq \bigcup_{\ell\in\mathcal{S}}\zon{\hat{\mathbf{x}}^{\ell}_{t|t}}{\mathbf{E}^{\ell}_{t|t}}$, so the event $\big\{\mathcal{X}_t(\mathbf{h}_{t|t})\subseteq \bigcup_{\ell\in\mathcal{S}}\zon{\hat{\mathbf{x}}^{\ell}_{t|t}}{\mathbf{E}^{\ell}_{t|t}}\big\}$ implies the event in \eqref{eqn:reduced_enclosure}. Hence, by monotonicity of the probability measure and Corollary~\ref{cor:prob},
\begin{multline}
\mathbb{P}\bigg\{\mathcal{X}_t(\mathbf{h}_{t|t})\subseteq \bigcup_{j\in \tilde{\mathcal{S}}}
\zon{\tilde{\mathbf{x}}^{j}_{t|t}}{\tilde{\mathbf{E}}^{j}_{t|t}} \, \big| \, \mathbf{h}_{t|t}\in\mathcal{A}_{t|t}\bigg\} \\
\geq \sum_{\ell\in\mathcal{S}} \gamma_{t|t}^{\ell} = \sum_{j\in\tilde{\mathcal{S}}}\sum_{\ell\in\mathcal{I}_j} \gamma_{t|t}^{\ell} = \sum_{j\in\tilde{\mathcal{S}}} \tilde{\gamma}_{t|t}^{j}, \notag 
\end{multline}
where the first equality holds since the index sets $\{\mathcal{I}_j\}$ are disjoint, and the second follows from the second condition in \eqref{eqn:zm_reduction_conditions}. \hfill $\square$
\end{pf}

Beyond preserving the enclosure guarantees, an admissible reduction also yields a valid zonotopic mixture description of the state, so that the recursion of Theorem~\ref{thm:zmf} may continue from the reduced mixture. This is a consequence of the following result.

\vspace{-0.15cm}
\begin{proposition}
\label{prop:zm_containment}
Let $\mathbf{w}\in\mathbb{R}^n$ follow the zonotopic mixture $\{(\alpha_i,\langle \mathbf{c}_i, \mathbf{E}_i \rangle)\}_{i=1}^{M}$, and let $\{(\tilde{\alpha}_j,\langle \tilde{\mathbf{c}}_j, \tilde{\mathbf{E}}_j \rangle)\}_{j=1}^{m}$ be an admissible reduction of it. Then, $\mathbf{w}$ also follows the reduced zonotopic mixture $\{(\tilde{\alpha}_j,\langle \tilde{\mathbf{c}}_j, \tilde{\mathbf{E}}_j \rangle)\}_{j=1}^{m}$.
\end{proposition}
\vspace{-0.15cm}

\begin{pf}
Let $m^\star$ denote the mode index of the original mixture, and define the random index $\tilde{m}\in\{1,\dots,m\}$ by $\tilde{m}=j$ if and only if $m^\star\in\mathcal{I}_j$, which is well defined since $\{\mathcal{I}_j\}_{j=1}^{m}$ is a partition. Then, $\mathbb{P}\{\tilde{m}=j\} = \sum_{i\in\mathcal{I}_j} \mathbb{P}\{m^\star=i\} = \tilde{\alpha}_j$. Conditioned on $m^\star=i$, the realization \eqref{eqn:zm_realization} yields $\mathbf{w}=\mathbf{c}_i+\mathbf{E}_i\mathbf{s}_i\in \langle\mathbf{c}_i, \mathbf{E}_i\rangle \subseteq \langle\tilde{\mathbf{c}}_{\tilde{m}},\tilde{\mathbf{E}}_{\tilde{m}}\rangle$ by \eqref{eqn:zm_reduction_conditions}. Thus, the representation \eqref{eqn:zm_realization} for the reduced mixture holds for $\mathbf{w}$. \hfill $\square$
\end{pf}

Admissible reductions can be constructed from elementary operations. An \emph{elementary merge} replaces two components $(\alpha_i, \langle \mathbf{c}_i, \mathbf{E}_i \rangle)$ and $(\alpha_j, \langle \mathbf{c}_j, \mathbf{E}_j \rangle)$ of a zonotopic mixture by the single component $(\alpha_i+\alpha_j, \langle \mathbf{c}_{ij}, \mathbf{E}_{ij} \rangle)$, where $\langle \mathbf{c}_{ij}, \mathbf{E}_{ij} \rangle \supseteq \langle \mathbf{c}_i, \mathbf{E}_i \rangle \cup \langle \mathbf{c}_j, \mathbf{E}_j \rangle$, leaving all other components unaltered. A constructive choice of the enclosing zonotope is given by \cite{girard2005reachability} as
\begin{equation}
\mathbf{c}_{ij}\hspace{-0.05cm}=\hspace{-0.05cm}\frac{\mathbf{c}_i\hspace{-0.06cm}+\hspace{-0.04cm}\mathbf{c}_j}{2},\hspace{-0.18cm}\quad \mathbf{E}_{ij}\hspace{-0.05cm}=\hspace{-0.05cm} \frac{1}{2}\hspace{-0.03cm}\left[\mathbf{E}_{i}\hspace{-0.05cm}+\hspace{-0.05cm}\mathbf{E}_{j}, \mathbf{E}_{i}\hspace{-0.07cm}-\hspace{-0.05cm}\mathbf{E}_{j}, \mathbf{c}_{i}\hspace{-0.07cm}-\hspace{-0.05cm}\mathbf{c}_{j} \right]\hspace{-0.02cm},
    \label{eqn:girard_merge}
\end{equation}
where the thinner generator matrix among $\mathbf{E}_i$ and $\mathbf{E}_j$ is padded with zero columns to match their dimensions. Repeated instantiation of elementary merges with the enclosure \eqref{eqn:girard_merge} shows that admissible reductions of any target size can be directly constructed. However, admissible reductions are non-unique, and repeated merging of distant zonotopes may yield overly conservative results. The remainder of this section addresses a selection of admissible reductions that minimize a suitable notion of information loss.

\vspace{-0.15cm}
\subsection{Greedy reduction algorithm}
\label{sec:loss}
\vspace{-0.15cm}
We quantify the information lost by an admissible reduction as the expected increase in the size of the component associated with the realized mode, i.e.,
\begin{equation}
    D:= \sum_{j=1}^{m}\sum_{i\in\mathcal{I}_j} \alpha_i\big[\rho(\langle \tilde{\mathbf{c}}_j, \tilde{\mathbf{E}}_j \rangle)-\rho(\langle \mathbf{c}_i, \mathbf{E}_i \rangle)\big],
    \label{eqn:reduction_loss}
\end{equation}
where $\rho(\cdot)$ is a size measure which in this paper is defined as being proportional to the mean width \cite{joos2023isoperimetric} of a zonotope weighted by $\mathbf{W}\succ \mathbf{0}$:
\begin{equation}
    \rho\big(\langle\mathbf{c},\mathbf{E}\rangle\big)
    := \sum_{k=1}^{p}\big\|\mathbf{e}_k\big\|_{\mathbf{W}},
    \label{eqn:rho_def}
\end{equation}
where $\mathbf{e}_k$ denotes the $k$th column of $\mathbf{E}\in\mathbb{R}^{n\times p}$. Although this choice is not unique in the proposed framework, it is motivated by it being strictly increasing with respect to proper set inclusion \cite{kabluchko2018monotonicity}, and computable in only $O(n^2 p)$ operations.

Minimizing $D$ jointly over partitions and enclosures is a combinatorial problem. Instead, a greedy scheme is implemented, which repeatedly performs the elementary merge incurring the smallest loss, in analogy with reduction schemes for Gaussian mixtures \cite{runnalls2007kullback}. Since all unaltered components cancel in the difference of weighted sums, an elementary merge of the components yields a closed form expression for a pairwise cost, presented in Lemma~\ref{lem:cost_closed_form}.

\begin{lemma}
\label{lem:cost_closed_form}
Consider an elementary merge of the components $(\alpha_i, \langle \mathbf{c}_i, \mathbf{E}_i \rangle)$ and $(\alpha_j, \langle \mathbf{c}_j, \mathbf{E}_j \rangle)$ into $(\alpha_i+\alpha_j, \langle \mathbf{c}_{ij}, \mathbf{E}_{ij} \rangle)$, with $\langle \mathbf{c}_{ij}, \mathbf{E}_{ij} \rangle$ given by \eqref{eqn:girard_merge}, where the thinner generator matrix is padded with zero columns so that both have $p$ columns. Then, the merge increases the loss \eqref{eqn:reduction_loss} by exactly
\begin{multline}
    C(i,j) \hspace{-0.04cm}=\hspace{-0.04cm} \frac{\alpha_i\hspace{-0.04cm}+\hspace{-0.04cm}\alpha_j}{2}
    \Bigg[\sum_{k=1}^{p}\hspace{-0.05cm}\Big(\|\mathbf{e}_{i,k}+\mathbf{e}_{j,k}\|_{\mathbf{W}}
      +\|\mathbf{e}_{i,k}-\mathbf{e}_{j,k}\|_{\mathbf{W}}\hspace{-0.03cm}\Big) \\
    +\hspace{-0.02cm}\|\mathbf{c}_i\hspace{-0.05cm}-\hspace{-0.04cm}\mathbf{c}_j\|_{\mathbf{W}}\hspace{-0.02cm}\Bigg]
    \hspace{-0.08cm}-\hspace{-0.05cm}\alpha_i\hspace{-0.03cm}\sum_{k=1}^{p}\|\mathbf{e}_{i,k}\|_{\mathbf{W}}
    \hspace{-0.04cm}-\hspace{-0.04cm}\alpha_j\sum_{k=1}^{p}\|\mathbf{e}_{j,k}\|_{\mathbf{W}},
    \label{eqn:cost_closed_form}
\end{multline}
where $\mathbf{e}_{i,k}$ and $\mathbf{e}_{j,k}$ denote the $k$th columns of $\mathbf{E}_i$ and $\mathbf{E}_j$, respectively. Furthermore, $C(i,j)\geq 0$ for every pair of components, with $C(i,j)=0$ if and only if $\langle \mathbf{c}_{ij}, \mathbf{E}_{ij} \rangle=\langle \mathbf{c}_i, \mathbf{E}_i \rangle=\langle \mathbf{c}_j, \mathbf{E}_j \rangle$.
\end{lemma}
\begin{pf}
Since all unaltered components cancel in \eqref{eqn:reduction_loss}, the increase in cost equals
\begin{multline}
        C(i,j)=(\alpha_i+\alpha_j)\rho(\langle \mathbf{c}_{ij}, \mathbf{E}_{ij} \rangle) \\
        -\alpha_i \rho(\langle \mathbf{c}_i, \mathbf{E}_i \rangle) - \alpha_j \rho(\langle \mathbf{c}_j, \mathbf{E}_j \rangle),
        \label{eqn:pairwise_cost}
\end{multline}
and \eqref{eqn:cost_closed_form} follows by applying \eqref{eqn:rho_def} to the generator matrix in \eqref{eqn:girard_merge}. For the second claim, rewrite \eqref{eqn:pairwise_cost} as
\begin{multline*}
    C(i,j)=\alpha_i\big[\rho(\langle \mathbf{c}_{ij}, \mathbf{E}_{ij} \rangle)-\rho(\langle \mathbf{c}_i, \mathbf{E}_i \rangle)\big]\\
    +\alpha_j\big[\rho(\langle \mathbf{c}_{ij}, \mathbf{E}_{ij} \rangle)-\rho(\langle \mathbf{c}_j, \mathbf{E}_j \rangle)\big].
\end{multline*}
Since $\langle \mathbf{c}_{ij}, \mathbf{E}_{ij} \rangle$ contains both $\langle \mathbf{c}_i, \mathbf{E}_i \rangle$ and $\langle \mathbf{c}_j, \mathbf{E}_j \rangle$, both bracketed terms are nonnegative by the monotonicity of $\rho$ under set inclusion \cite{kabluchko2018monotonicity}, and hence $C(i,j)\geq 0$. Moreover, as $\alpha_i,\alpha_j>0$, $C(i,j)=0$ holds if and only if both bracketed terms vanish, which by the equality case of Lemma~2.1 of \cite{kabluchko2018monotonicity} occurs if and only if $\langle \mathbf{c}_{ij}, \mathbf{E}_{ij} \rangle=\langle \mathbf{c}_i, \mathbf{E}_i \rangle=\langle \mathbf{c}_j, \mathbf{E}_j \rangle$. \hfill $\square$
\end{pf}

In the ZMF, Algorithm~\ref{alg:zm_reduction} is applied after each measurement update, with each merged component carrying the sum of the weights of its constituents and the test \eqref{eq:ctest} applied once per merged component. The generator order reduction in Step~7 exploits $\zon{\mathbf{c}}{\mathbf{E}}\subseteq\zon{\mathbf{c}}{\redu{q,\mathbf{W}}\mathbf{E}}$ to keep the generator matrices tractable. Since the merged components enclose their constituents and the ZMF updates preserve set inclusion, Corollaries~\ref{cor:prob} and~\ref{cor:reduced_prob} remain valid at every time step with the merged weights, zonotopes, and survivor set in place of the ones given by Theorem~\ref{thm:zmf}. 

\begin{algorithm}[t]
\caption{Greedy zonotopic mixture reduction}
\label{alg:zm_reduction}
\begin{algorithmic}[1]
\Statex \textbf{Input:} Mixture $\{(\alpha_i,\langle\mathbf{c}_i,
  \mathbf{E}_i\rangle)\}_{i=1}^{M}$ with $\alpha_i>0$, target size $m$, generator
  budget $q$, weighting $\mathbf{W}\succ \mathbf{0}$.
\State $n \gets M$
\While{$n > m$}
  \State Compute $C(i,j)$ in \eqref{eqn:cost_closed_form} for $i<j\leq n$
  \State $(i^\star,j^\star) \gets \arg\min_{i<j\leq n} C(i,j)$
  \State Merge $\langle\mathbf{c}_{i^\star},\mathbf{E}_{i^\star}\rangle$ and $\langle\mathbf{c}_{j^\star},\mathbf{E}_{j^\star}\rangle$ via \eqref{eqn:girard_merge}
  \If{$\mathbf{E}_{i^\star j^\star}$ has more than $q$ generators}
    \State Reduce the generators of $\mathbf{E}_{i^\star j^\star}$ via \eqref{eq:reducedgenerator}
  \EndIf
  \State $n \gets n-1$
\EndWhile
\Statex \textbf{Output:} Reduced mixture $\{(\tilde{\alpha}_j,\langle
  \tilde{\mathbf{c}}_j,\tilde{\mathbf{E}}_j\rangle)\}_{j=1}^{m}$.
\end{algorithmic}
\end{algorithm}

\begin{rem}
The cost \eqref{eqn:cost_closed_form} of absorbing a component of negligible weight $\alpha_i$ into a nearby component $j$ is approximately $\alpha_j[\rho(\langle \mathbf{c}_{ij}, \mathbf{E}_{ij} \rangle)-\rho(\langle \mathbf{c}_j, \mathbf{E}_j \rangle)]$, which is small whenever the absorbed component lies near $\langle \mathbf{c}_j, \mathbf{E}_j \rangle$. The greedy scheme therefore eliminates low-weight components preferentially, emulating the pruning step of Gaussian sum filters \cite{crouse2011look}.
\end{rem}

\section{Simulation studies} 
\label{sec:simulation}

Two examples are presented to showcase the properties and effectiveness of the ZMF. The first is a fourth order system in which we examine the falsification test \eqref{eq:ctest}, the enclosure probability property of Corollary~\ref{cor:prob}, and performance against the ZKF and GSF. The second is a three area load frequency control benchmark, on which pruning and the ZMR scheme of Algorithm \ref{alg:zm_reduction} are compared over a Monte Carlo simulation.
	
\subsection{Example 1: Fourth-order system}\label{sec:ex1}

Consider the system \eqref{eqn:system} with $n_x=4$, $n_u=2$, $n_y=2$ and block diagonal matrices $\mathbf{A}=\textnormal{diag}\{ \mathbf{A}_a,\mathbf{A}_b\}$,  $\mathbf{B}=\textnormal{diag}\{ \mathbf{B}_a,\mathbf{B}_b\}$, $\mathbf{C}=\textnormal{diag}\{ \mathbf{C}_a,\mathbf{C}_b\}$, and $\mathbf{D}=\textnormal{diag}\{1, 0.5\}$, with
\begin{alignat}{3}
    \mathbf{A}_a &\!=\!\!\begin{bmatrix} 0.950 & 0.100\\ -0.080 & 0.900\end{bmatrix}\!,~
   &\mathbf{B}_a &\!=\!\!\begin{bmatrix} 1.200\\ 1.100\end{bmatrix},~
   &\mathbf{C}_a^{\top} &\!=\!\!\begin{bmatrix} 1.000 \\ 0.000\end{bmatrix}, \notag\\
    \mathbf{A}_b &\!=\!\!\begin{bmatrix} 0.052 & -0.295\\ 0.295 & 0.052\end{bmatrix}\!,~
   &\mathbf{B}_b &\!=\!\!\begin{bmatrix} 0.111\\ -0.011\end{bmatrix},~
   &\mathbf{C}_b^{\top} &\!=\!\!\begin{bmatrix} 0.848 \\ -0.530\end{bmatrix}. \notag
\end{alignat}
The process noise is a zonotopic mixture of $W=10$ modes, built as the product of a $2$ mode mixture on the upper block and a $5$ mode mixture on the lower one. The upper factor has weights $[\alpha^a_1,\alpha^a_2]=[0.6,0.4]$, centers $\bm{\mu}^a_1\hspace{-0.04cm}=\hspace{-0.04cm}-\bm{\mu}^a_2=[1.2,1.1]^\top$, and generators $\mathbf{Q}^a_1\hspace{-0.04cm}=\hspace{-0.04cm}\mathbf{Q}^a_2\hspace{-0.04cm}=\hspace{-0.04cm}\operatorname{diag}\left\lbrace 0.05,0.05\right\rbrace $. The lower factor has weights, centers and generator matrices
\begin{align}
    [\alpha^b_1,\dots,\alpha^b_5] &= [0.50,0.23,0.14,0.08,0.05], \notag\\
    [\bm{\mu}^b_1,\dots,\bm{\mu}^b_5] &= \begin{bmatrix}
    0 & -0.331 & -0.635 & -1.432 & 2.057\\
    0 & 0.678 & -2.375 & -1.405 & 1.839\end{bmatrix}, \notag
\end{align}
\vspace{-0.9cm}
\begin{align}
    \mathbf{Q}^b_1 \!&=\!\!\begin{bmatrix} .117 & .090 & -.010\\ -.044 & .074 & .088\end{bmatrix}\!, &
    \mathbf{Q}^b_2 \!&=\!\!\begin{bmatrix} .149 & -.020 & .002\\ .015 & .134 & -.135\end{bmatrix}\!, \notag\\
    \mathbf{Q}^b_3 \!&=\!\!\begin{bmatrix} .234 & .215 & .175\\ -.257 & .215 & .071\end{bmatrix}\!, &
    \mathbf{Q}^b_4 \!&=\!\!\begin{bmatrix} .343 & .174 & .063\\ .101 & -.275 & .332\end{bmatrix}\!, \notag\\
    \mathbf{Q}^b_5 \!&=\!\!\begin{bmatrix} .431 & .047 & .354\\ -.087 & .446 & .116\end{bmatrix}\!. \notag
\end{align}
 The mode $(i,j)$ of the joint mixture has weight $\alpha^a_i\alpha^b_j$, center $[\bm{\mu}_i^{a\top},\bm{\mu}_j^{b\top}]^\top$ and generator $\operatorname{diag}\left\lbrace \mathbf{Q}^a_i,\mathbf{Q}^b_j \right\rbrace$. 

The output noise is built in the same way and has $V=4$ modes. Its upper factor has weights $[\eta_1^a,\eta_2^a]=[0.7,0.3]$, centers $\rho^a_1=-\rho^a_2=1.5$, and generators $R^a_1=R^a_2=0.1$, while its lower factor has weights $[\eta_1^b,\eta_2^b]=[0.6,0.4]$, centers $\rho^b_1=-\rho^b_2=0.12$, and generators $R^b_1=R^b_2=0.33$. The initial state zonotope is centered at the origin with generator matrix $\operatorname{diag}\left\lbrace \mathbf{E}_a,\mathbf{E}_b \right\rbrace $, where $\mathbf{E}_a=\operatorname{diag}\left\lbrace 0.3,0.3\right\rbrace$ and $\mathbf{E}_b=[0.297,\,0.186;\,-0.186,\,0.297]$. The implementation uses the greedy reduction of Algorithm~\ref{alg:zm_reduction} with target reduction size of $m=700$ and generator budget $q=20$.

We first study the number of surviving histories after the falsification step \eqref{eq:ctest} and ZMR. Figure \ref{fig:Example1Branches} shows that falsification via \eqref{eq:ctest} significantly trims the surviving branch count, even before ZMR is applied. Exhaustive mode enumeration exceeds $10^{32}$ branches after $t=20$, while the target reduction size $m$ caps it at $700$, and falsification via \eqref{eq:ctest} drives it below the upper limit whenever the data rules out histories, reaching a minimum of $289$ branches at $t=8$.

\begin{figure}[t]
    \centering
    \includegraphics[width=\columnwidth]{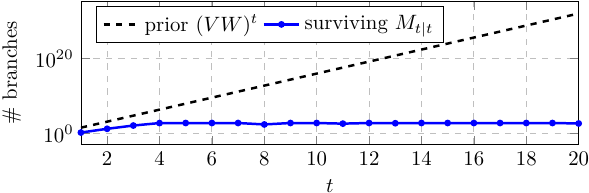}
    \vspace{-0.6cm}
    \caption{Example 1. Surviving branches $M_{t|t}$ of the posterior zonotopic mixture against the number of mode histories $(VW)^t=40^t$ that an exhaustive enumeration considers.}
    \label{fig:Example1Branches}
\end{figure}

Figure~\ref{fig:Example1Enclosure} draws the posterior mixture in the $(x_1,x_2)$ plane at six instants of the run, together with the true trajectory and a point state estimate, obtained as the weighted average center of is associated zonotopic mixture. The $4$ candidate combinations of the process and output mode of this block are separated enough for \eqref{eq:ctest} to reject all of them but the realized one, so the posterior holds a single zonotope of unit weight and the ZMF reduces to the ZKF of Lemma~\ref{lem:zkf} running on the true mode sequence. Since only one zonotope describes the state membership, ZMR is not required and the enclosure of Theorem~\ref{thm:zmf} holds without the additional conservatism of Corollary~\ref{cor:reduced_prob}.

\begin{figure}[t]
    \centering
    \includegraphics[width=\columnwidth]{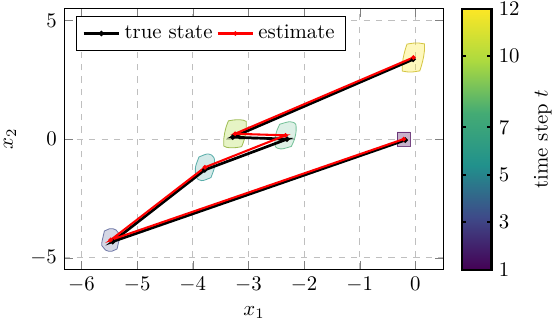}
    \vspace{-0.7cm}
    \caption{Example 1. Posterior mixture in the $(x_1,x_2)$ plane at $t\in\{1,3,5,7,10,12\}$, colored by time, with the true trajectory and its center point estimate. At every $t$ the mixture holds only one component, and the true state is inside it.}
    \label{fig:Example1Enclosure}
\end{figure}

We now compare the estimates of $x_1$ and $x_2$ over time obtained with the ZMF, the ZKF, and the GSF. The ZKF runs on the single zonotope enclosing each noise mixture by repeated application of \eqref{eqn:girard_merge}, while the GSF is computed with a Gaussian mixture model with the same weights and means as the zonotopic mixture noise, and covariances matching the corresponding covariation matrices. Figure~\ref{fig:Example1StateEstimation} shows the guaranteed bounds of the ZMF and ZKF, with their point estimates as lines. Both enclosures contain the true state at all instants by construction, but the ZKF band is $13.2$ and $5.7$ times wider than the ZMF band on $x_1$ and $x_2$. The ZKF cannot discard falsified mode histories; in contrast, the ZMF has identified the realized history and carries only its generators. The GSF follows the point estimate of the ZMF, with root mean square errors of $0.032$ and $0.105$ against $0.030$ and $0.104$, but it reports a density of unbounded support instead of a set, so it certifies no enclosure at all.

\begin{figure}[t]
    \centering
    \includegraphics[width=\columnwidth]{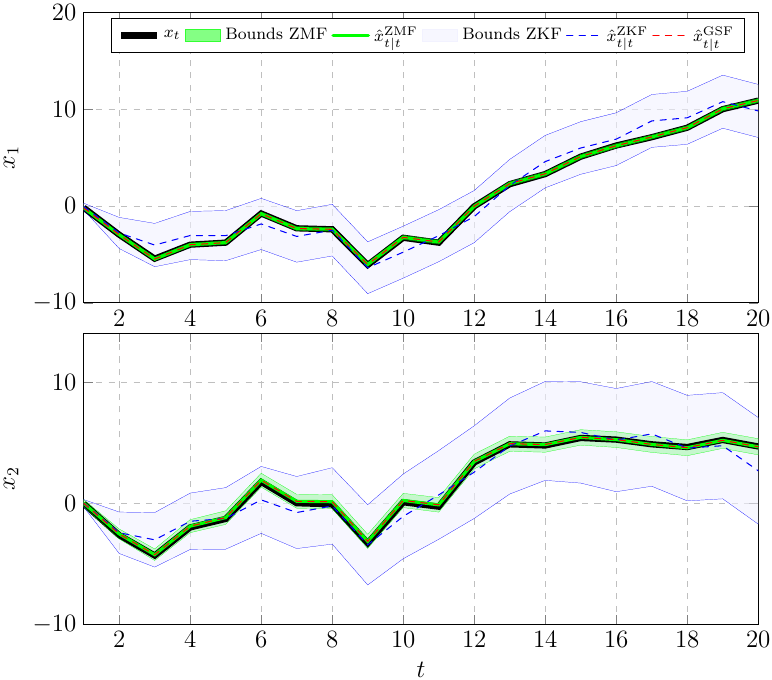}
    \vspace{-0.6cm}
    \caption{Example 1. True state, point estimates and guaranteed bounds of two states. Both the ZMF and the ZKF contain the truth state, and the ZKF enclosure, built on a single zonotope per noise, is up to an order of magnitude wider.}
    \label{fig:Example1StateEstimation}
\end{figure}

In agreement with Corollary~\ref{cor:prob}, the ZMF allows for selecting subsets of mode histories with a guaranteed enclosure probability. Figure~\ref{fig:Example1ProbLine} plots, at each instant, the area of the sub-union of the heaviest components whose weights add up to a guaranteed probability $\beta$, normalized by the area of the full union. This is the trade certified by Corollary~\ref{cor:prob} at $t=2$ and $t=3$, and by Corollary~\ref{cor:reduced_prob} from $t=4$ onward, where the reduction has merged components. The areas are computed by integrating over the boundary of the union, so the overlapping zonotopic regions are automatically discounted. Only the accumulated weight values are attainable values of $\beta$, so the reported area is piecewise constant in $\beta$ whose steps are as wide as the weight of the component that enters the mixture. The area-enclosure probability tradeoff is accentuated close to $\beta=1$, where reducing a small part of the probability guarantee implies a significant reduction of the reported area.

\begin{figure}[t]
    \centering
    \includegraphics[width=\columnwidth]{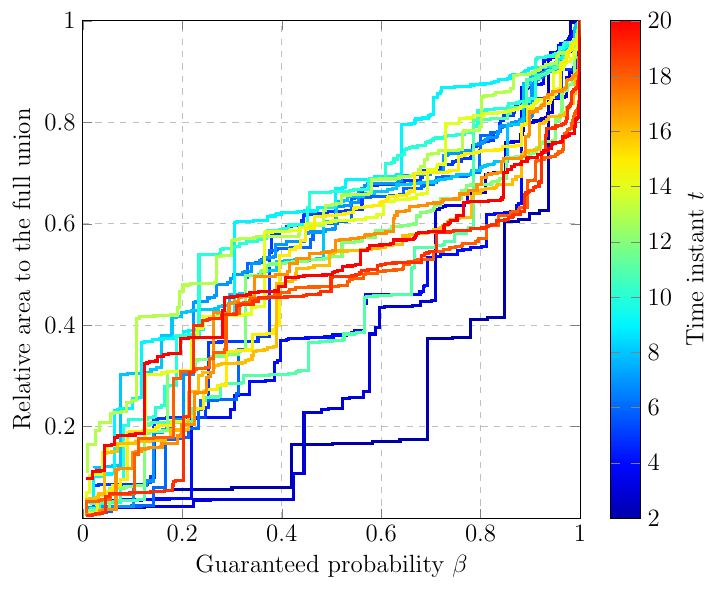}
    \vspace{-0.5cm}
    \caption{Example 1. Area of the sub-union of heaviest components of the upper block with guaranteed probability $\beta$, relative to the area of the full union, for $t=2$ to $t=20$. Color encodes the instant.}
    \label{fig:Example1ProbLine}
\end{figure}

\begin{figure*}[t]
    \centering
    \includegraphics[width=\textwidth]{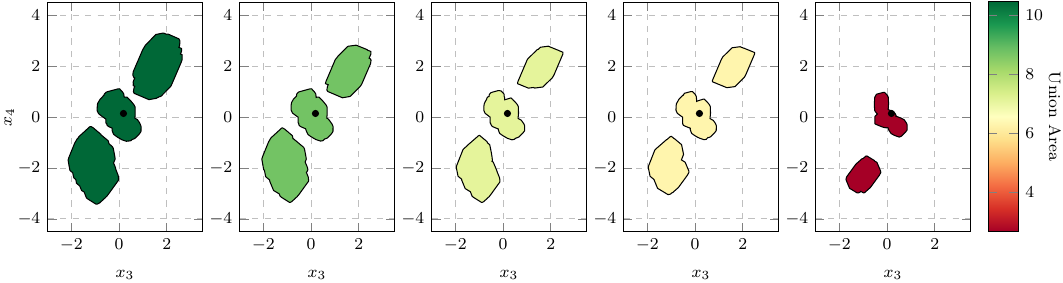}
    \vspace{-0.7cm}
    \caption{Example 1. Sub-unions of the upper block at $t=6$ for guaranteed probabilities $\beta=1.00$, $0.90$, $0.70$, $0.50$ and $0.30$, built respectively from $700$, $448$, $240$, $123$ and $51$ of the $700$ surviving components. Color encodes the area of each sub-union and the dot marks the true state.}
    \label{fig:Example1SetStar}
\end{figure*}
Figure~\ref{fig:Example1SetStar} shows one mixture enclosure at $t=6$, where it holds its $700$ components. From left to right the panels lower $\beta$ from $1.00$ to $0.30$, with enclosure areas corresponding to $100$, $84$, $68$, $60$ and $25$ percent of the full union. Already at full guaranteed probability the reported region is disconnected, which the ZKF cannot deliver. The three disconnected regions survive down to $\beta=0.50$, and giving up the next twenty points of confidence removes the region in the upper right corner, leaving a quarter of the original area. In this realization the true state stays inside the central region at all five levels, although only the first is guaranteed to contain the true state.

\subsection{Example 2: three-area load-frequency control}
\label{sec:ex2}

The second example is a three-area load-frequency control benchmark of interconnected power systems \cite{elgerd1970optimum,kundur1994power}. Each area contains an identical non-reheat thermal unit with states $\Delta f_i,\Delta P_{g,i}$, and $\Delta X_{g,i}$ (frequency deviation, generated power, and governor valve position), with power system gain $K_p=120$ Hz/pu, droop $R=2.4$ Hz/pu and governor, turbine, and power system time constants $T_g=0.08$ s, $T_t=0.3$ s, and $T_p=20$ s. The areas are interconnected radially by two tie lines of synchronizing coefficient $2\pi T_{ij}=0.545$ pu/rad, whose power deviations $\Delta P_{\textnormal{tie},12}$ and $\Delta P_{\textnormal{tie},23}$ complete the state vector of $n_x=11$. The input has dimension $n_u=6$, consisting of three governor setpoints and the three scheduled load changes, and the output has dimension $n_y=5$, consisting of three area frequencies and two tie line flows. The continuous time model is discretized by a zero-order hold at $T_s=0.25$ s and excited by a staircase input over $N=40$ samples.

The noise follows Definition~\ref{def:zonotopic_mixture}, with the mode histories representing fault sequences. The process noise has $W=7$ modes: nominal operation, with $\bm{\mu}_1=\mathbf{0}$ and $\alpha_1=0.70$, and six unmeasured load-altering events, one per area at either $0.045$ pu or $0.090$ pu, each with weight $0.05$. All modes share a generator matrix bounding the residual load fluctuation at $0.0023$ pu per area. The output noise has $V=3$ modes: a healthy sensor set, with $\bm{\rho}_1=\mathbf{0}$ and $\eta_1=0.7$; a $0.15$ Hz bias on the frequency meter of Area~1 with $\eta_2=0.18$; and a $0.06$ pu bias on the flow meter of the first tie line, with $\eta_3=0.12$. All outputs share a bounded sensor noise of $0.0046$ Hz on the frequency channels and $0.0023$ pu on the tie line channels.

The initial state is contained in a zonotope centered at the origin, whose generators give half-widths of $0.03$ in the area states and $0.015$ in the tie line states. The generator budget is $q=200$ and the target reduction size is $m=6$, which is deliberately smaller than the number of process modes. Every time update creates $WM_{t|t}$ branches with no data available to falsify them, so the predicted mixture exceeds the budget and the reduction is exercised at every step of the recursion, no matter how many branches survive the measurement update. Four estimators are compared over $100$ Monte Carlo realizations. The first two are the ZMF reduced by the ZMR scheme of Algorithm~\ref{alg:zm_reduction} and by weight pruning, which retains the $6$ heaviest components and discards the rest. The third is the GSF induced by \eqref{gmmnoise} with covariances $\mathbf{Q}_i\mathbf{Q}_i^\top$ and $\mathbf{R}_j\mathbf{R}_j^\top$ under the same component budget, whose confidence band is taken as the $3\sigma$ box of the posterior mixture. The fourth is the ZKF of \cite{combastel2015zonotopes} run on the single zonotope that encloses each noise mixture.

Figure~\ref{fig:Example2Bands} reports the guaranteed bounds of the two reductions on one realization, for the frequency of Area~1 and for the flow of the first tie line. The bounds resulting from ZMR using Algorithm \ref{alg:zm_reduction} contain the true trajectory at every instant. The pruning bounds are similar to the greedy algorithm before $t=8$, but at $t=8$ pruning discards the branch carrying the realized mode history, which cannot be recovered. The subsequent data falsifies every surviving history, and no estimator can be reported from then on. In sharp contrast, reduction via merging replaces a group of components by a mixture that contains all of them, so no probability mass is lost and the guarantee of Corollary~\ref{cor:reduced_prob} is inherited by the reduced mixture. 

\begin{figure}[t]
    \centering
    \includegraphics[width=\columnwidth]{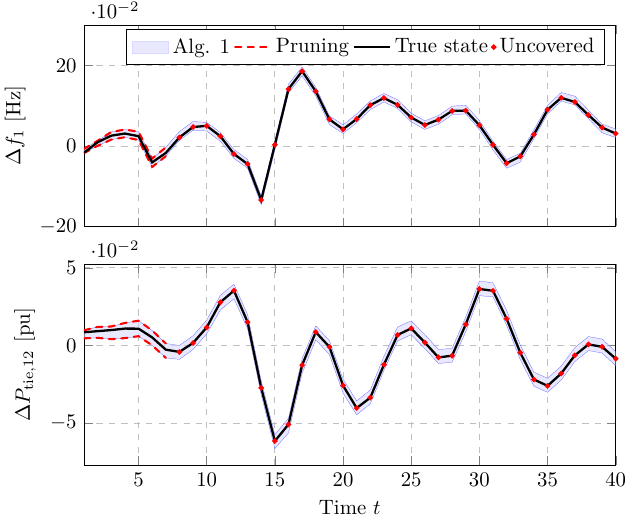}
    \vspace{-0.6cm}
    \caption{Example 2. Guaranteed bounds of the ZMF on one realization, for the frequency of area~1 and the flow of the first tie line. The greedy reduction encloses the true trajectory throughout. The pruning reduction is tighter and leaves the truth uncovered at the marked instants.}
    \label{fig:Example2Bands}
\end{figure}

\begin{table*}[t]
    \centering
    \footnotesize
    \setlength{\tabcolsep}{5pt}
    \caption{Example 2. Monte Carlo summary over $100$ realizations, with standard deviations in parentheses. Coverage is the fraction of instants where the true state lies in the reported enclosure, non-empty is the fraction of instants where the posterior mixture retains at least one component, $\bar{\rho}$ is the expected component size of \eqref{eqn:rho_def}, and $\bar{w}_{\beta}$ is the mean width on the frequency states of the sub-union of heaviest components of the corrected mixture with guaranteed probability $\beta$.}
    \label{tab:Example2}
    \begin{tabular}{lccccccc}
\hline
Filter & RMSE $\Delta f$ [Hz] & Coverage [\%] & Non-empty [\%] & $\bar{\rho}$ & $\bar{w}_{1}$ [Hz] & $\bar{w}_{0.9}$ [Hz] & Time [s] \\
\hline
ZMF, Alg. \ref{alg:zm_reduction} & 0.0018 (0.0001) & 100.0 & 100.0 & 0.036 & 0.0214 & 0.0214 & 8.119 \\
ZMF, pruning & 0.0020 (0.0004) & 47.1 & 47.4 & 0.051 & 0.0197 & 0.0197 & 2.989 \\
GSF & 0.0050 (0.0018) & 82.0 & 100.0 & 0.059 & 0.0187 & 0.0186 & 0.054 \\
ZKF & 0.0255 (0.0020) & 100.0 & 100.0 & 0.368 & 0.1839 & 0.1839 & 0.031 \\
\hline
\end{tabular}

\end{table*}

Table~\ref{tab:Example2} summarizes the Monte Carlo simulation. The ZMF with greedy ZMR attains the lowest point error, $0.0018$ Hz against the $0.0050$ Hz of the GSF and the $0.0255$ Hz of the ZKF, and achieves it with full coverage. The GSF replaces each bounded mode by a Gaussian of matching covariance, so it spreads weight over histories that the consistency test \eqref{eq:ctest} rejects, and its $3\sigma$ band contains the true state in only $82$ percent of the instants against $100$ percent of the ZMF. The ZKF also achieves full coverage, and at $0.031$ s per run against the $8.12$ s of the greedy ZMF, is by far the fastest, but its expected component size $\bar{\rho}$ is ten times that of the ZMF and it attaches no probability to any part of its enclosure. Lastly, pruning is the cheapest reduction method, although it cannot guarantee enclosure.

\vspace{-0.2cm}
\section{Conclusions}
\label{sec:conclusions}
\vspace{-0.2cm}
This paper has proposed zonotopic mixture models, in which the noise is generated by drawing a zonotope from a finite collection according to fixed probabilities and then realizing an arbitrary element of it. For state estimation, this noise description has led to the zonotopic mixture filter, which admits an interpretation as a bank of zonotopic Kalman filters over mode histories, with posterior probabilities computed from the data. Guaranteed coverage probabilities follow directly from the recursions, and a greedy mixture reduction scheme was proposed that retains these guarantees, keeping the representation computationally feasible. Two numerical examples illustrated the benefits of the approach: the consistency test discards falsified histories and reduces the branch count by orders of magnitude, and zonotope unions with predefined enclosure probabilities can be computed, even with merging-based reduction.

Fault detection and diagnosis is a natural application of this framework, since each mixture mode can be interpreted as a fault and the posterior weights quantify the evidence for each fault sequence. Future work concerns the derivation of a zonotopic mixture smoother, extensions to nonlinear systems and simultaneous input and state estimation, and applications of the ZMF within control and system identification.

\vspace{-0.2cm}
\section*{Appendix A: Proof of Lemma \ref{lem:pmf}}
\label{appendix:pmf}
\vspace{-0.2cm}
\begin{pf}
    The initial law $P(\mathbf{x}_1|\mathbf{y}_{1:0})$ follows directly from $P(\mathbf{x}_1)$. Proceeding by induction, assume that \eqref{pmf_prediction} holds for all time instants up to $t-1$. Then, the Bayesian measurement update recursion \cite{sarkka2013bayesian} yields
    \begin{align}
    P&(\mathbf{x}_t|\mathbf{y}_{1:t}) \notag \\
    &\hspace{-0.2cm}\propto P(\mathbf{y}_t|\mathbf{x}_t)P(\mathbf{x}_t|\mathbf{y}_{1:t-1}) \notag \\
    &\hspace{-0.2cm}\propto \sum_{j=1}^V \sum_{k=1}^{M_{t|t-1}} \eta_j \gamma_{t|t-1}^k \delta(\mathbf{y}_t, \mathbf{Cx}_{t} + \mathbf{Du}_t \hspace{-0.07cm}+\hspace{-0.07cm} \bm{\rho}_j\hspace{-0.01cm}) \delta(\mathbf{x}_{t},\hat{\mathbf{x}}_{t|t-1}^k) \notag \\
    &\hspace{-0.2cm}= \sum_{j=1}^V \sum_{k=1}^{M_{t|t-1}} \eta_j \gamma_{t|t-1}^k \delta(\bm{\nu}_t^\ell,\mathbf{0}) \delta(\mathbf{x}_{t},\hat{\mathbf{x}}_{t|t-1}^k). \notag
    \end{align}
    Defining the index $\ell$ as in the lemma statement and normalizing the PMF yields \eqref{eqn:filtering_pmf}. On the other hand, the time update equations follow from
    \begin{align}
        P(\mathbf{x}_{t+1}|\mathbf{y}_{1:t}) \hspace{-0.05cm} = \hspace{-0.05cm} \sum_{j=1}^V \sum_{\ell=1}^{M_{t|t}} \alpha_i \gamma_{t|t}^\ell \delta(\mathbf{x}_{t+1},\mathbf{A}\hat{\mathbf{x}}_{t|t}^\ell+\mathbf{Bu}_t +\bm{\mu}_i), \notag 
    \end{align}
    where we have used the fact that $\sum_{\mathbf{x}_t\in \mathcal{X}_t}\delta(\mathbf{x}_t,\hat{\mathbf{x}}_{t|t}^\ell) = 1$. Defining the new index $k$ as in the lemma statement yields \eqref{pmf_prediction}, concluding the proof.   
\end{pf}

\section*{Appendix B: Proof of Theorem \ref{thm:filters-dbf}}
\begin{pf}
\emph{Part (i)}. We proceed by induction, starting with the proportionality of the generator matrices. As the ZMF initializes with generator $\varepsilon \mathbf{E}_1$, we only need to show that $\mathbf{E}_{t|t}^\ell$ and $\mathbf{E}_{t+1|t}^k$ are proportional to $\varepsilon$ given that $\mathbf{E}_{t|t-1}^k=\varepsilon \tilde{\mathbf{E}}_{t|t-1}^k$, with $\tilde{\mathbf{E}}_{t|t-1}^k$ independent of $\varepsilon$. Defining $\tilde{\mathbf{P}}_{t|t-1}^k:= \tilde{\mathbf{E}}_{t|t-1}^k\tilde{\mathbf{E}}_{t|t-1}^{k\top}$ and replacing it in \eqref{optimalgainzmf}, the common factor $\varepsilon^2$ cancels:
\begin{equation}
    \mathbf{L}_t^\ell = \tilde{\mathbf{P}}_{t|t-1}^k \mathbf{C}^\top(\mathbf{C} \tilde{\mathbf{P}}_{t|t-1}^k\mathbf{C}^\top + \mathbf{R}_j \mathbf{R}_j^\top)^{-1} \notag,
\end{equation}
so $\mathbf{L}_t^\ell$ is independent of $\varepsilon$. Thus, $\mathbf{E}_{t|t}^\ell = \varepsilon [(\mathbf{I}-\mathbf{L}^\ell_t \mathbf{C})\tilde{\mathbf{E}}^k_{t|t-1}, \mathbf{L}^\ell_t\mathbf{R}_j]$, and since the process zonotope generators $\varepsilon \mathbf{Q}_i$ are also proportional to $\varepsilon$, the same follows for $\mathbf{E}_{t+1|t}^k$ from \eqref{generator_timeupdate}.

Define $\bar r:=\max\{\|[\mathbf{C}\tilde{\mathbf{E}}^{k}_{t|t-1},\mathbf{R}_j]\|_{2,1}\colon \ell \leq M_{t|t},t\leq N\}$ and $\underline{\nu}=\min\{\|\bm{\nu}^{\ell,\mathrm{DBF}}_t\|_2\colon \ell\leq M_{t|t},t\leq N, \bm{\nu}^{\ell,\mathrm{DBF}}_t\neq \mathbf{0}\}$, where the superscript DBF denotes quantities computed by the DBF, and analogously for ZMF. Both quantities are finite and positive, and we set $\varepsilon^*=\underline{\nu}/\bar r$. 

Fix $\varepsilon\in(0,\varepsilon^*)$. Since the time updates for $\gamma_{t+1|t}^k$ and $\hat{\mathbf{x}}_{t+1|t}^k$ are identical in both filters, and both start from the single center $\bar{\mathbf{x}}_1$, it suffices to show that the measurement updates keep the unnormalized weights equal branchwise and the positive-weight centers equal to the DBF atoms. Zero-weight branches stay at zero weight in both filters since the updates are multiplicative. On a branch with $\gamma^{k,\mathrm{ZMF}}_{t|t-1}>0$, the center equality gives $\bm{\nu}^{\ell,\mathrm{ZMF}}_t=\bm{\nu}^{\ell,\mathrm{DBF}}_t$. If $\bm{\nu}^{\ell,\mathrm{DBF}}_t=\mathbf{0}$, then the indicator $c_t^\ell =1$, which implies $\bar{\gamma}^{\ell,\mathrm{ZMF}}_{t|t}=\bar{\gamma}^{\ell,\mathrm{DBF}}_{t|t}$, and the center update of the ZMF $\hat{\mathbf{x}}^{k,\mathrm{ZMF}}_{t|t-1}+\mathbf{L}^{\ell}_t\mathbf{0}$ leaves the center unchanged, in agreement with \eqref{eqn:hatxfiltering}. If $\bm{\nu}^{\ell,\mathrm{DBF}}_t\neq\mathbf{0}$, then every point of the innovation zonotope $\zon{\mathbf{0}}{\varepsilon[\mathbf{C}\tilde{\mathbf{E}}^{k}_{t|t-1},\mathbf{R}_j]}$ has norm at most $\varepsilon\bar r<\underline{\nu}\leq\|\bm{\nu}^{\ell,\mathrm{DBF}}_t\|_2$, so the indicator is zero, matching the DBF. The unnormalized weights thus coincide on every branch, and as the realized branch has positive weight, the common normalizing constant is positive and the normalized weights coincide.

\emph{Part (ii)}. The state and innovation covariance of the GSF do not depend on the branch since every mixture component shares the covariances $\bm{\Sigma}^w$ and $\bm{\Sigma}^v$. Thus, we write $\tilde{\mathbf{P}}_{t|t}$ and $\tilde{\mathbf{S}}_t$ without branch superscripts. Induction from $\varepsilon^2\bm{\Sigma}_1$ yields $\mathbf{P}_{t|t}=\varepsilon^2\tilde{\mathbf{P}}_{t|t}$ and $\mathbf{S}_{t}=\varepsilon^2\tilde{\mathbf{S}}_{t}$, with $\tilde{\mathbf{P}}_{t|t}, \tilde{\mathbf{S}}_{t}\succ \mathbf{0}$ independent of $\varepsilon$. As a consequence, iterating the weight updates gives, for a branch $\ell$ at time $t$,
\begin{align}
\label{gammaexpanded}
    \gamma^{\ell,\mathrm{GSF}}_{t|t}&=\frac{\pi(\mathbf{h}_{t|t}^\ell) e^{-q_\ell/(2\varepsilon^2)}}{\sum_{\ell'=1}^{M_{t|t}}\pi(\mathbf{h}_{t|t}^{\ell'}) e^{-q_{\ell'}/(2\varepsilon^2)}}, \\
    q_\ell:&=\sum_{\tau=1}^{t}(\bm{\nu}_\tau^{\ell,\mathrm{GSF}})^\top\tilde{\mathbf{S}}_\tau^{-1}\bm{\nu}^{\ell,\mathrm{GSF}}_\tau\geq 0. \notag
\end{align}
We claim that $q_\ell=0$ if and only if $\bm{\nu}^{\ell,\mathrm{DBF}}_\tau=\mathbf{0}$ for all $\tau\leq t$. To this end, we prove by induction that, for every $\tau\leq t$, if $\bm{\nu}^{\ell,\mathrm{GSF}}_r=\bm{\nu}^{\ell,\mathrm{DBF}}_r=\mathbf{0}$ for all $r<\tau$, then the GSF mean and the DBF atom coincide at time $\tau$, and consequently $\bm{\nu}^{\ell,\mathrm{GSF}}_\tau=\bm{\nu}^{\ell,\mathrm{DBF}}_\tau$. Indeed, both filters start from $\bar{\mathbf{x}}_1$ and share the time update \eqref{eqn:hat_x1_filtering_theorem}, and, at every $r<\tau$, the zero innovation makes the GSF measurement update $\hat{\mathbf{x}}^{k,\mathrm{GSF}}_{r|r-1}+\mathbf{K}_r^\ell\mathbf{0}$ coincide with \eqref{eqn:hatxfiltering}. Hence one innovation sequence vanishes at every $\tau\leq t$ if and only if the other does, which proves the claim.

Since $\mathbbm{1}\{q_\ell\hspace{-0.02cm}=\hspace{-0.02cm}0\}\hspace{-0.02cm}=\hspace{-0.02cm}\prod_{\tau\leq t}\delta(\bm{\nu}^{\ell,\mathrm{DBF}}_\tau,\mathbf{0})$ by the claim above,
\begin{equation}
    \gamma^{\ell,\mathrm{GSF}}_{t|t}\xrightarrow{\varepsilon\to 0}\frac{\pi(\mathbf{h}_{t|t}^\ell)\mathbbm{1}\{q_\ell=0\}}{\sum_{\ell'=1}^{M_{t|t}}\pi(\mathbf{h}_{t|t}^{\ell'})\mathbbm{1}\{q_{\ell'}=0\}}=\gamma^{\ell,\mathrm{DBF}}_{t|t}, \notag
\end{equation}
and the prediction weights $\gamma^{k}_{t+1|t}$ also converge, since their updates are the same. 

Finally, the GSF filtering posterior is the finite Gaussian mixture described by $\{(\gamma^{\ell,\mathrm{GSF}}_{t|t},\hat{\mathbf{x}}^{\ell,\mathrm{GSF}}_{t|t},\varepsilon^2\tilde{\mathbf{P}}_{t|t})\}_{\ell=1}^{M_{t|t}}$, where the means are independent of $\varepsilon$. For any bounded continuous $f$, each Gaussian component with fixed mean and vanishing covariance converges weakly to the point mass at its mean. Thus, since the sum is finite,
\begin{align}
    &\int_{\mathbb{R}^n} f(\mathbf{x})\sum_{\ell=1}^{M_{t|t}}\gamma^{\ell,\mathrm{GSF}}_{t|t}\mathcal{N}(\mathbf{x};\hat{\mathbf{x}}^{\ell,\mathrm{GSF}}_{t|t},\varepsilon^2\tilde{\mathbf{P}}_{t|t})\mathrm{d}\mathbf{x}\notag \\
    &\overset{\varepsilon\to 0}{\xrightarrow{\mkern20mu}}\hspace{-0.1cm}\sum_{\ell=1}^{M_{t|t}}\hspace{-0.12cm}\gamma^{\hspace{-0.02cm}\ell,\mathrm{DBF}}_{t|t} \hspace{-0.07cm}f(\hspace{-0.01cm}\hat{\mathbf{x}}^{\hspace{-0.02cm}\ell,\mathrm{GSF}}_{t|t}\hspace{-0.01cm}) \hspace{-0.09cm}=\hspace{-0.15cm} \int_{\mathbb{R}^n} \hspace{-0.25cm}f(\hspace{-0.01cm}\mathbf{x}\hspace{-0.01cm})\hspace{-0.1cm}\sum_{\ell=1}^{M_{t|t}}\hspace{-0.1cm}\gamma^{\hspace{-0.02cm}\ell,\mathrm{DBF}}_{t|t}\hspace{-0.05cm}\delta(\hspace{-0.01cm}\mathbf{x},\hspace{-0.025cm}\hat{\mathbf{x}}^{\hspace{-0.02cm}\ell,\mathrm{DBF}}_{t|t}\hspace{-0.01cm})\mathrm{d}\mathbf{x}, \notag
\end{align}
where the swap between $\hat{\mathbf{x}}^{\ell,\mathrm{GSF}}_{t|t}$ and $\hat{\mathbf{x}}^{\ell,\mathrm{DBF}}_{t|t}$ in the last equality is justified by the fact that these quantities are equal when $\gamma^{\ell,\mathrm{DBF}}_{t|t}>0$. Hence, the posterior of the GSF converges weakly to \eqref{eqn:filtering_pmf}, completing the proof. \hfill$\square$
\end{pf}

\balance
\bibliography{Bibliography}
\end{document}